\documentclass[11pt]{article}
\def\withcolors{1}
\def\withnotes{1}

\usepackage{amsmath,amsthm,enumitem,nicefrac,bbm,verbatim,amssymb,amsfonts,amscd, graphicx,algorithm,algpseudocode,hyperref,float,mathtools,bm,url,tikz}
\usetikzlibrary{shapes,arrows,chains}
\usetikzlibrary[calc]

\newtheorem{Theorem*}{Theorem}

\newtheorem{Claim*}[Theorem]{Claim}

\newtheorem{CounterExample*}{$\overline{\hbox{\bf Example}}$}

\newtheorem{Example*}[Theorem]{Example}

\newtheorem{Intuition*}[Theorem]{Intuition}
\newtheorem{Joke*}[Theorem]{Joke}

\newtheorem{Lemma*}[Theorem]{Lemma}
\newtheorem{Open problem}[Theorem]{Open problem}

\newtheorem{Question*}[Theorem]{Question}

\usepackage{fullpage}

\newtheorem{theorem}{Theorem}
\newtheorem{proposition}[theorem]{Proposition}

\newtheorem{lemma}[theorem]{Lemma}

\newtheorem{definition}[theorem]{Definition}

\def \bSubexa    {\begin{subexa}}

\usepackage[colorinlistoftodos,textsize=scriptsize]{todonotes}

\newcommand{\ignore}[1]{}

\newcommand{\EE}{\mathbb{E}}

\newcommand{\RR}{\mathbb{R}}

\def \cC     {{\cal C}}

\def \cF     {{\cal F}}

\def \cL     {{\cal L}}

\def \cN     {{\cal N}}

\newcommand{\ed}{\stackrel{\mathrm{def}}{=}}

\def\ignore#1{}

\newcommand{\bi}{\begin{itemize}}
\newcommand{\ei}{\end{itemize}}

\def\orpro{\mathop{\mathchoice
   {\vee\kern-.49em\raise.7ex\hbox{$\cdot$}\kern.4em}
   {\vee\kern-.45em\raise.63ex\hbox{$\cdot$}\kern.2em}
   {\vee\kern-.4em\raise.3ex\hbox{$\cdot$}\kern.1em}
   {\vee\kern-.35em\raise2.2ex\hbox{$\cdot$}\kern.1em}}\limits}

\def\andpro{\mathop{\mathchoice
 {\wedge\kern-.46em\lower.69ex\hbox{$\cdot$}\kern.3em}
 {\wedge\kern-.46em\lower.58ex\hbox{$\cdot$}\kern.25em}
 {\wedge\kern-.38em\lower.5ex\hbox{$\cdot$}\kern.1em}
 {\wedge\kern-.3em\lower.5ex\hbox{$\cdot$}\kern.1em}}\limits}

\def\simge{\mathrel{%
   \rlap{\raise 0.511ex \hbox{$>$}}{\lower 0.511ex \hbox{$\sim$}}}}

\def\simle{\mathrel{
   \rlap{\raise 0.511ex \hbox{$<$}}{\lower 0.511ex \hbox{$\sim$}}}}

\providecommand{\email}[1]{\href{mailto:#1}{\nolinkurl{#1}\xspace}}

\ifnum\withcolors=1
  \newcommand{\newest}[1]{{\color{orange} {#1}}} 
  \newcommand{\acolor}[1]{{\color{purple}#1}} 
  \newcommand{\bcolor}[1]{{\color{cyan}#1}} 
  
\else

  \newcommand{\newest}[1]{{{#1}}}
  
  \newcommand{\acolor}[1]{{#1}}
  \newcommand{\bcolor}[1]{{#1}}
\fi

\ifnum\withnotes=1
  \newcommand{\anote}[1]{\par\acolor{\textbf{JA: }\sf #1}} %
  \newcommand{\bnote}[1]{\par\bcolor{\textbf{SB: }\sf #1}} %
  \newcommand{\wnote}[1]{\par\scolor{\textbf{AW: }\sf #1}} %

\else
  \newcommand{\anote}[1]{}
  \newcommand{\bnote}[1]{}
  \newcommand{\wnote}[1]{}

\fi

\newcommand{\eps}{\varepsilon}

\newcommand{\pcnd}[3][]{
\ifthenelse{\equal{#1}{}}{p\left(#2 \middle| #3\right)}{p_{#1} \left(#2 \middle| #3 \right)}
}
\newcommand{\qcnd}[3][]{
\ifthenelse{\equal{#1}{}}{q\left(#2 \middle| #3\right)}{q_{#1} \left(#2 \middle| #3 \right)}
}
\newcommand{\qjnt}[3][]{
\ifthenelse{\equal{#1}{}}{q\left(#2, #3\right)}{q_{#1} \left(#2 , #3 \right)}
}
\newcommand{\pjnt}[3][]{
\ifthenelse{\equal{#1}{}}{p\left(#2, #3\right)}{p_{#1} \left(#2,#3 \right)}
}
\newcommand{\qmrg}[2][]{
\ifthenelse{\equal{#1}{}}{q\left(#2 \right)}{q_{#1} \left(#2 \right)}
}
\newcommand{\pmrg}[2][]{
\ifthenelse{\equal{#1}{}}{p\left(#2 \right)}{p_{#1} \left(#2 \right)}
}

\DeclarePairedDelimiterX{\infdivx}[2]{(}{)}{%
  #1\|#2%
}

\newcommand{\expc}[2][]{\ifthenelse{\equal{#1}{}} {\mathbb{E}\left[ #2 \right] }{\mathbb{E}_{#1} \left[ #2 \right]}}

\newcommand{\ipdim}{d}

\newcommand{\tr}{\text{tr}}
\newcommand{\cov}{\bm{K}}
\newcommand{\dcov}{\bm{\Lambda}}
\newcommand{\encnew}{\bm{W}}
\newcommand{\encvec}{\bm{w}}
\newcommand{\decnew}{\bm{T}}

\newcommand{\ort}{\bm{Q}}
\newcommand{\ortvec}{\bm{q}}
\newcommand{\ide}{\bm{I}}
\newcommand{\encv}[2][]{\ifthenelse{\isempty{#1}}{\bm{w}_{#2}}{w_{#1 #2}}}

\newcommand{\rotencv}[2][]{\ifthenelse{\isempty{#1}}{\bm{y}_{#2}}{y_{#1 #2}}}

\DeclareMathOperator*{\argmin}{arg\,min}

\newcommand{\data}{\bm{x}}
\newcommand{\latent}{\bm{y}}

\newcommand{\diagcov}{\bm{\Sigma}}

\newcommand{\distortion}{D}
\newcommand{\scalePBA}{\alpha}

\newcommand{\noise}{\bm{\eps}}
\newcommand{\noisevar}{\sigma^2}
\newcommand{\rateproxy}{\rho}
\newcommand{\eigvec}{\bm{u}}
\newcommand{\eigmat}{\bm{U}}
\newcommand{\diagenc}{\widetilde{\bm{W}}}
\newcommand{\var}{\bm{y}}

\newcommand{\vort}{\bm{Y}}

\newcommand{\send}{r}

\newcommand{\scaling}{s}
\newcommand{\scalemat}{\bm{S}}

\newcommand{\genmat}{\bm{M}}
\newcommand{\genmatvec}{\bm{m}}
\newcommand{\genvec}{\bm{v}}
\newcommand{\lag}{\cL}
\newcommand{\setD}{\left \lbrace D_i \right \rbrace_{i=1}^{d}}
\newcommand{\setmu}{\left \lbrace \mu_i  \right \rbrace_{i=1}^{d}}
\newcommand{\orthvar}{\mu}

\usepackage{utopia}
\usepackage{xifthen}
\usepackage{float}

\title{Principal Bit Analysis: \\
Autoencoding with Schur-Concave Loss}

\author{
Sourbh Bhadane\thanks{This research was supported by the US National Science Foundation
under grants CCF-2008266, CCF-1934985, CCF-1617673, CCF-1846300, CCF-1815893 and the US
Army Research Office under grant W911NF-18-1-0426.	
}\\
Cornell University\\
\tt{snb62@cornell.edu}\\
\and
Aaron B. Wagner\footnotemark[1]\\
Cornell University\\
\tt{wagner@cornell.edu}\\
\and
Jayadev Acharya\footnotemark[1]\\
Cornell University\\
\tt{acharya@cornell.edu}
}

\begin{document}
\maketitle
\begin{abstract}
We consider a linear autoencoder in which the latent variables
    are quantized, or corrupted by noise, and the constraint
    is Schur-concave in the set of latent variances.
    Although finding the optimal encoder/decoder pair 
    for this setup is a nonconvex optimization problem, we show that 
    decomposing the source into its principal components is optimal. If 
    the constraint is strictly Schur-concave and the empirical
	  covariance matrix has only simple eigenvalues, then any optimal
	  encoder/decoder must decompose the source in this way. As one application,
     we consider a strictly Schur-concave constraint that estimates the 
      number of bits needed to represent the latent variables under 
      fixed-rate encoding, a setup that
    we call \emph{Principal Bit Analysis (PBA)}.  This yields a practical,
    general-purpose, fixed-rate compressor that outperforms 
    existing algorithms. As a second application, we show that a
      prototypical autoencoder-based variable-rate compressor is guaranteed
	  to decompose the source into its principal components.

\end{abstract}

\section{Introduction}

\emph{Autoencoders} are an effective method for representation learning
and dimensionality reduction. Given a centered dataset $\data_1, \data_2, \ldots, \data_n
\in \mathbb{R}^d$ (i.e., $\sum_i \data_i = 0$),
an autoencoder (with \emph{latent dimension} $k \le d$)
consists of an \emph{encoder} $f:
\mathbb{R}^d \mapsto \mathbb{R}^k$ and a \emph{decoder} $g: \mathbb{R}^k \mapsto
\mathbb{R}^d$. The goal is to select $f$ and $g$ from prespecified
classes $\mathcal{C}_f$ and $\mathcal{C}_g$ respectively such that if a random point $\data$ is picked from the data set then $g(f(\data))$ is close to $\data$ in some sense, for example in mean squared error. 
If $\mathcal{C}_f$ and $\mathcal{C}_g$ consist of linear mappings then the autoencoder is called a \emph{linear autoencoder}.

Autoencoders have achieved striking successes when $f$ and $g$ are
selected through training from the class of functions realized
by multilayer perceptrons of a given architecture~\cite{HintonS06}.
Yet, the canonical autoencoder formulation described above has
a notable failing, namely that for linear autoencoders,
optimal choices of $f$ and $g$ do not necessarily identify the 
principal components of the dataset; they merely identify the principal
subspace~\cite{BourlardK88, BaldiH89}. That is, the components of $f(\data)$ 
are not necessarily proportional to projections of $\data$ against the eigenvectors of the covariance matrix
\begin{equation}\label{eq:cov}
\cov\ed \frac{1}{n} \sum_{i = 1}^n \data_i \cdot \data_i^\top,
\end{equation}
which we assume without loss of generality is full rank.
Thus, linear autoencoders do not recover Principal Component
Analysis (PCA). The reason {for this} is that both the objective 
(the distortion) and the constraint (the dimensionality
of the latents) are invariant to an invertible transformation 
applied after the encoder with its inverse applied 
before the decoder. It is desirable for linear autoencoders to recover PCA for two reasons. First, from a representation learning standpoint, it guarantees that the autoencoder recovers uncorrelated features. Second, since a conventional linear autoencoder has a large number of globally optimal solutions corresponding to different bases of the principal subspace, it is preferable to eliminate this indeterminism.

Autoencoders are sometimes described as ``compressing'' the
data~\cite{Santo12,BourlardK88, LiaoZWLL21, Bishop06}, even though $f$
can be invertible even when $k < d$.
We show that by embracing this {compression-}view, one can obtain 
autoencoders that are {able} to recover PCA.
Specifically, we consider linear autoencoders with quantized (or,
equivalently, noisy) latent variables with a constraint on the estimated
number of bits required to transmit the quantized latents
under fixed-rate coding. We call this problem 
\emph{Principal Bit Analysis (PBA).}
The constraint turns out to be a strictly Schur-concave
function of the set of variances of the latent variables (see 
the supplementary for a review of Schur-concavity). Although
finding the optimal $f$ and $g$ for this loss function is
a nonconvex optimization problem,  we
show that for any strictly Schur-concave loss function,
an optimal $f$ must send projections of the data along the 
principal components, assuming that the empirical covariance
matrix of the data has only simple eigenvalues. That is,
imposing a strictly Schur-concave loss in place of a
simple dimensionality constraint suffices to ensure recovery of PCA.
The idea is that the strict concavity of the loss function 
eliminates the rotational invariance described above.
As we show, even a slight amount of ``curvature'' in the 
constraint forces the autoencoder to spread the variances
of the latents out as much as possible, resulting in recovery
of PCA. If the loss function is merely Schur-concave, then 
projecting along the principal components is optimal, but 
not necessarily uniquely so. 

Using this theorem, we can efficiently solve PBA.
We validate the solution experimentally by using it to construct
a fixed-rate compression algorithm for arbitrary vector-valued
data sources.  We find that the PBA-derived compressor  beats
existing linear, fixed-rate compressors both in terms of
mean squared error, for which it is optimized, and in terms
of the structural similarity index measure (SSIM) and downstream
classification accuracy, for which it is not. 

A number of variable-rate multimedia compressors have recently
been proposed that are either related to, or directly
inspired by, autoencoders \cite{TschannenAL18, TodericiVJHMSC17,
BalleLS16, TodericiOHVMBCS16, TheisSCH17, RippelB17, HabibianRTC19,
AgustssonMTCTBG17, BalleMSHJ18, ZhouCGSW18, AgustssonTMTG19, BalleCMSJAHT20}. 
As a second application of our result,
we show that for Gaussian sources, a linear form of such a
compressor is guaranteed to recover PCA. Thus we show that
ideas from compression can be fruitfully fed back into the
original autoencoder problem.

The contributions of the paper are
\begin{itemize}
\item We propose a novel linear autoencoder formulation in which
	the constraint is Schur-concave. We show that this generalizes
      conventional linear autoencoding.
\item If the constraint is strictly Schur-concave and the 
        covariance matrix of the data has only simple 
       eigenvalues, then we show that the autoencoder
       provably recovers PCA, providing a new remedy
       for a known limitation of linear autoencoders.
\item We use the new linear autoencoder formulation to efficiently
	  solve a fixed-rate compression problem that we
       call \emph{Principal Bit Analysis (PBA).}
\item We demonstate experimentally that PBA outperforms
	 existing fixed-rate compressors on a variety
		of data sets and metrics.
\item We show that a linear, variable-rate compressor that 
      is representative of many autoencoder-based
	compressors in the literature effectively has
       a strictly Schur-concave loss, and therefore it
		recovers PCA.
\end{itemize}






\textbf{Related Work.} Several recent works have examined how linear autoencoders can be 
modified to guarantee recovery of PCA. Most solutions involve eliminating the invariant global optimal solutions by introducing regularization of some kind. \cite{OftadehSWS20} propose a loss function which adds $k$ penalties to recover the $k$ principal directions, each corresponding to recovering up to the first $i \leq k$ principal directions. \cite{KuninBGS19} show that $\ell_2$ regularization helps reduce the symmetry group to the orthogonal group. \cite{BaoLSG20} further break the symmetry by considering non-uniform $\ell_2$ regularization and deterministic dropout. \cite{LadjalNP19} consider a nonlinear autoencoder with a covariance loss term to encourage finding orthogonal directions. Recovering PCA is an important problem even in the stochastic counterpart of autoencoders. \cite{LucasTGN19} analyze linear variational autoencoders (VAEs) and show that the global optimum of its objective is identical to the global optimum of log marginal likelihood of probabilistic PCA (pPCA). \cite{RolinekZM19} analyze an approximation to the VAE loss function and show that the linear approximation to the decoder is orthogonal. 

Our result on variable-rate compressors is connected to the sizable recent
literature on compression using autoencoder-like architectures. Representative 
contributions to the literature were noted above. Those works focus mostly
on the empirical performance of deep, nonlinear networks, with a particular
emphasis on finding a differentiable proxy for quantization so as to train
with stochastic gradient descent. In contrast, this work considers provable
properties of the compressors when trained perfectly.

\textbf{Notation.} We denote matrices by bold capital letters e.g. $\genmat$, and vectors by bold small, e.g. $\genvec$. The $j^{\text{th}}$ column of a matrix $\genmat$ is denoted by $\genmatvec_j$ and the $j^{\text{th}}$ entry of a vector $\genvec$ by $\left[\genvec\right]_j$. We denote the set $\left \lbrace 1,2, \cdots \ipdim \right \rbrace$ by $\left[ \ipdim \right]$. A sequence $a_1, a_2, \cdots a_n$ is denoted by $\{a_i\}_{i=1}^{n}$. We denote the zero column by $\bm{0}$. Logarithms without specified bases denote natural logarithms.

\textbf{Organization.} The balance of the paper is organized as follows. We describe our constrained linear autoencoder framework in Section~\ref{sec:FRAMEWORK}. This results in an optimization problem that we solve for any Schur-concave constraint in Section~\ref{subsec:PROOF}. In Section~\ref{sec:PBAPCA}, we recover linear autoencoders and PBA under our framework. We apply the PBA solution to a problem in variable-rate compression of Gaussian sources in Section~\ref{sec:VARIABLE}. Section~\ref{sec:EXPERIMENTS} contains experiments comparing the performance of the PBA-based fixed-rate compressor against existing fixed-rate linear compressors on image and audio datasets.



\section{Linear Autoencoding with a Schur-Concave Constraint}\label{sec:FRAMEWORK}

Throughout this paper we consider $\mathcal{C}_f$ and 
$\mathcal{C}_g$ to be the class of linear functions. The 
functions $f\in\mathcal{C}_f$ and $g \in\mathcal{C}_g$ can then be represented by $d$-by-$d$ matrices, 
respectively, which we denote by
$\encnew$ and $\decnew$, respectively. Thus we have
\begin{align}
    f(\data) & = \encnew^\top \data \\
    g(\data) & = \decnew \data.
\end{align}
We wish to design $\encnew$ and $\decnew$ to
minimize the mean squared error when the latent variables
$\encnew^\top \data$ are quantized, subject to a constraint
on the number of bits needed to represent the quantized latents.
We accomplish this via two modifications of the canonical autoencoder.
First, we perturb the $d$ latent variables with zero-mean
additive noise with covariance matrix $\noisevar \ide$, which we denote by
$\noise.$ Thus the input to the decoder is 
\begin{equation}
\encnew^\top \data + \noise
\end{equation}
and our objective is to minimize the mean squared error
\begin{equation}
    \label{eq:1stmse}
 \frac{1}{n} \sum\limits_{i=1}^{n} \EE_{\noise}\left[ \left \lVert \data_i - \decnew\left( \encnew^\top\data_i+ \noise \right) \right \rVert_2^2 \right].
\end{equation}

This is equivalent to quantizing the latents, in the following sense~\cite{ZamirF92}.
Let $Q(\cdot)$ be the function that maps any real number to its nearest 
integer and $\eps$ be a random variable uniformly distributed over $[-1/2,1/2]$. Then for $X$ independent of $\eps$, the quantities $Q(X + \eps) - \eps$
and $X + \eps$ have the same joint distribution with $X$.
Thus (\ref{eq:1stmse}) is exactly the mean squared error if the latents
are quantized to the nearest integer and $\noisevar = \frac{1}{12}$, 
assuming that the quantization is dithered. The overall system is depicted
in Fig.~\ref{fig:blockdiag}.

%

\tikzstyle{block} = [rectangle, draw, text width=5em, text centered, rounded corners, minimum height=4em]
\tikzstyle{line} = [draw, -latex']
\tikzstyle{input} = [coordinate]
\begin{figure}[!htb]
\begin{center}
\begin{tikzpicture}[node distance=1cm, auto]
	\node [input, name=input] {};
    \node [block,right = 1.5cm of input] (A) {Linear Encoder ($\encnew$)};
    \node [block, right=1.5cm of A] (B) {Quantizer};
    \node [block, right=2cm of B] (C) {Linear Decoder ($\decnew$)};
    \node [text width = 1.5cm, right = 0.5cm of C] (D) {$\decnew\left(\encnew^\top \data_i + \bm{\eps} \right)$};
	\draw [draw,->] (input) -- node {$\data_i$} (A);
    \path [line] (A) -- node [text width=2cm,midway,above,align=center ] {$\encnew^\top \data_i$} (B);
    \path [line] (B) -- node [text width=2cm,midway,above,align=center ] {$\encnew^\top \data_i + \bm{\eps}$} (C);
    \path [line] (C) -- (D);

\end{tikzpicture}
\end{center}
\caption{Compression Block Diagram}\label{fig:blockdiag}
\end{figure}
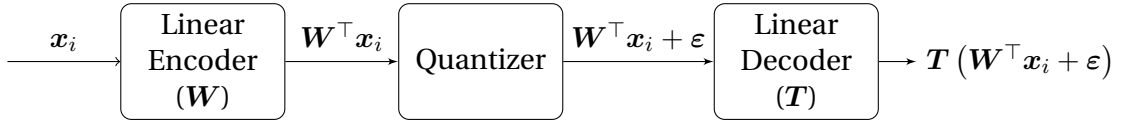
We wish to constrain the number of bits needed to describe
the latent variables. We assume that the $j$th quantized latent
is clipped to the interval 
$$\left(-\frac{\sqrt{ (2a)^2 \encvec_j^\top \cov \encvec_j + 1}}{2},
\frac{\sqrt{(2a)^2 \encvec_j^\top \cov \encvec_j + 1}}{2}\right],$$ where
$a > 0$ is a hyperparameter and the covariance matrix $\cov$ is as defined in \eqref{eq:cov}. The idea is that for sufficiently
large $a$, the interval $$\left(-a\sqrt{ \encvec_j^\top \cov 
\encvec_j}, a\sqrt{ \encvec_j^\top \cov \encvec_j}\right]$$ 
contains the latent with high probability, and adding $1$
accounts for the expansion due to the dither. The number of
bits needed for the $j$th latent is then
\begin{align}
    \log \left(\sqrt{ 4a^2 \encvec_j^\top \cov \encvec_j + 1}\right)
      = \frac{1}{2} \log \left(4a^2 \encvec_j^\top \cov \encvec_j + 1\right).
\end{align}
We arrive at our optimization problem:
\begin{equation}\label{eq:MASTEROPTIPBA}
\begin{aligned}
\inf\limits_{\encnew,\decnew}  \quad & \frac{1}{n} \sum\limits_{i=1}^{n} \EE_{\noise}\left[ \left \lVert \data_i - \decnew\left( \encnew^\top\data_i+ \noise \right) \right \rVert_2^2 \right] \\ 
    \text{subject to} \quad & R \geq \sum_{j = 1}^d \frac{1}{2} \log \left(4a^2 \encvec_i^\top \cov \encvec_i + 1\right).
\end{aligned}
\end{equation}
Note that the function 
$$\{\encvec_j^\top \cov \encvec_j\}_{j = 1}^d \mapsto \sum_{j = 1}^d \frac{1}{2} \log \left(4a^2 \encvec_i^\top \cov \encvec_i + 1\right)$$
is strictly Schur-concave (see Appendix~\ref{app:schur} for a brief review
of Schur-concavity). Our first result only requires that the constraint
is Schur-concave in the set of latent variances, so we will consider
the more general problem

\begin{equation}\label{eq:MASTEROPTI}
\begin{aligned}
\inf\limits_{\encnew,\decnew}  \quad & \frac{1}{n} \sum\limits_{i=1}^{n} \EE_{\noise}\left[ \left \lVert \data_i - \decnew\left( \encnew^\top\data_i+ \noise \right) \right \rVert_2^2 \right] \\ 
    \text{subject to} \quad & R \geq \rho\left(\left\{\encvec_j^\top \cov \encvec_j\right\}_{j = 1}^d\right)
\end{aligned}
\end{equation}
where $\rateproxy(\cdot)$ is any Schur-concave function.

Expressing the objective in~\eqref{eq:MASTEROPTI} in terms of $\cov$, the optimization problem reduces to
\begin{equation}\label{eq:MASTEROPTI_LIN}
\begin{aligned}
\inf \limits_{\encnew, \decnew}  \quad & \tr\left( \cov \right) - 2\tr\left(\cov \encnew \decnew^\top \right) + \tr\left( \decnew\left( \encnew^\top \cov \encnew + \noisevar\bm{I} \right) \decnew^\top \right) \\ 
\text{subject to} \quad & R \geq \rho\left(\left\{\encvec_j^\top \cov \encvec_j\right\}_{j = 1}^d\right).
\end{aligned}
\end{equation}

Since $\decnew$ does not appear in the rate constraint, the optimal $\decnew$ can be viewed as the Linear Least Squares Estimate (LLSE) of a random $\data$ given $\encnew^\top \data + \noise$. Therefore, the optimal decoder, $\decnew^*$ for a given encoder $\encnew$ is (e.g. \cite{Kay98}):

\begin{equation}
\decnew^* = \cov \encnew (\encnew^{\top} \cov \encnew + \noisevar\bm{I})^{-1}.
\end{equation}
Substituting for $\decnew$ in \eqref{eq:MASTEROPTI_LIN} yields an optimization problem over only $\encnew$
\begin{equation} \label{eq:MASTEROPTI_MMSE}
\begin{aligned} 
\inf\limits_{\encnew} \quad & \tr(\cov) - \tr(\cov \encnew (\encnew^{\top} \cov \encnew + \noisevar\bm{I} )^{-1}\encnew^{\top} \cov ) \\
\text{subject to} \quad & R \geq \rho\left(\left\{\encvec_j^\top \cov \encvec_j\right\}_{j = 1}^d\right).
\end{aligned}
\end{equation}

This problem is nonconvex in general.
In the following subsection, we prove a structural result about
the problem for a Schur-concave $\rateproxy$. Namely, we show that
the nonzero rows of $\encnew$ must be eigenvectors of $\cov$.
In Section~\ref{sec:PBAPCA}, we solve the problem for the specific 
choice of $\rateproxy$ in (\ref{eq:MASTEROPTIPBA}). We also show how this
generalizes conventional linear autoencoders.

\subsection{Optimal Autoencoding with a Schur-Concave Constraint}\label{subsec:PROOF}


The following is the main theoretical result of the paper.



\begin{theorem}\label{thm:diagisoptimal}
For Schur-concave $\rateproxy:\RR_{\geq 0}^{d} \rightarrow \RR_{\geq 0}$ and $R>0$, the set of matrices whose nonzero columns are eigenvectors of the covariance matrix $\cov$ is optimal for \eqref{eq:MASTEROPTI_MMSE}. If $\rateproxy$ is strictly Schur-concave and $\cov$ contains distinct eigenvalues, this set contains all optimal solutions of \eqref{eq:MASTEROPTI_MMSE}.
\end{theorem}

\begin{proof}


Let the eigenvalues of $\cov$ be $\{\sigma_i^2\}_{i=1}^{d}$ with $\sigma_1^2 \ge \sigma_2^2 \ge \ldots \ge \sigma^2_d$. Let the eigendecomposition of $\cov$ be given by $\cov = \eigmat \diagcov \eigmat^{\top}$ where $\eigmat$ is an orthogonal matrix whose columns are the eigenvectors of $\cov$ and $\diagcov$ is a diagonal matrix with entries $\left \lbrace \sigma_i^2 \right \rbrace_{i=1}^{d}$.

We first prove that the optimal value of \eqref{eq:MASTEROPTI_MMSE} can be achieved by \newest{a} $\encnew$ such that $\encnew^{\top} \cov \encnew$ is a diagonal matrix. Let $\diagenc = \encnew \ort$ where $\ort$ is the orthogonal matrix obtained from the eigendecomposition of $\encnew^{\top} \cov \encnew$ i.e., $$\encnew^{\top}\cov \encnew = \ort\dcov \ort^{\top},$$ where $\dcov$ is a diagonal matrix formed from the eigenvalues of $ \encnew^{\top}\cov \encnew$. Note that 
\begin{align*}
\tr \left( \cov \widetilde{\encnew} \left( \widetilde{\encnew}^{\top} \cov \widetilde{\encnew} + \noisevar\ide \right)^{-1} \widetilde{\encnew}^{\top} \cov\right) &= \tr \left( \cov \encnew \ort \left(  \dcov + \noisevar\ide \right)^{-1} \ort^{\top} \encnew^{\top} \cov \right) \\
&= \tr \left(  \cov \encnew \left( \ort \dcov \ort^{\top} + \noisevar\ort \ort^{\top} \right)^{-1} \encnew^{\top} \cov \right).
\end{align*}
Since $\ort \dcov \ort^{\top} = \encnew^{\top} \cov \encnew$ and $\ort \ort^{\top}= \ide$, the objective remains the same. 
We now show that the constraint is only improved. Denoting the eigenvalues of $\encnew^{\top} \cov \encnew$ by $\left \lbrace \nu_j \right \rbrace_{j=1}^d$, we have
\begin{equation*}
\rateproxy\left( \left \lbrace \widetilde{\encvec}_j^{\top} \cov \widetilde{\encvec}_j \right \rbrace_{j=1}^{\ipdim} \right) 
=  \rateproxy\left( \left \lbrace \ortvec_j^{\top} \encnew^{\top} \cov \encnew \ortvec_j \right \rbrace_{j=1}^{\ipdim} \right) =  \rateproxy \left( \left \lbrace  \nu_j \right \rbrace_{j=1}^{\ipdim} \right).
\end{equation*}
%
    
Now since the eigenvalues of a Hermitian matrix majorize its diagonal elements by the Schur-Horn theorem \cite[Theorem~4.3.45]{HornJ85},
$$ \left \lbrace \encvec_j^{\top} \cov \encvec_j \right \rbrace_{j=1}^{d} \prec \left \lbrace \nu_j \right \rbrace_{j=1}^{\ipdim}.$$
Since $\rateproxy$ is Schur-concave, this implies
$$ \rateproxy\left( \left \lbrace \encvec_j^{\top} \cov \encvec_j \right \rbrace_{j=1}^{\ipdim} \right) \geq  \rateproxy \left( \left \lbrace \nu_j \right \rbrace_{j=1}^{\ipdim} \right) = \rateproxy \left( \left \lbrace \widetilde{\encvec}_j^{\top} \cov \widetilde{\encvec}_i \right \rbrace_{j=1}^{\ipdim} \right). $$
Therefore, if $\rateproxy$ is Schur-concave, the rate constraint can only improve. This implies an optimal solution can be attained when 
    $\encnew$ is such that $\encnew^{\top} \cov \encnew$ is diagonal. If $\rateproxy$ is strictly Schur-concave, the rate constraint strictly improves implying that the optimal $\encnew$ must be such that $\encnew^{\top} \cov \encnew$ is diagonal.
 This implies that
 \begin{align*}
\tr\left( \cov \encnew \left( \encnew^{\top} \cov \encnew + \noisevar\ide \right)^{-1} \encnew^{\top} \cov \right)  &= \tr\left(  \encnew^{\top} \cov^2 \encnew \left( \encnew^{\top} \cov \encnew + \noisevar\ide \right)^{-1}\right) \\
&= \sum\limits_{i=1}^{d} \frac{\encvec_i^{\top} \cov^2 \encvec_i}{\noisevar + \encvec_i^{\top} \cov \encvec_i}.
\end{align*}   

Note that minimizing the objective in \eqref{eq:MASTEROPTI_MMSE} is equivalent to maximizing the above expression. Perform the change of variable 
    \begin{align*}
        \encvec_j & \mapsto \begin{cases} 
            \left(\frac{\cov^{1/2} \encvec_j}{||\cov^{1/2} \encvec_j||}, ||\cov^{1/2} \encvec_j||^2\right) & \text{if $\cov^{1/2} \encvec_j \ne \bm{0}$} \\
        (\bm{0},0) & \text{if $\cov^{1/2} \encvec_j = \bm{0}$}
    \end{cases}  \\
        &  = (\var_j,y_j).
    \end{align*}
The assumption that $\encnew^\top \cov \encnew$ is diagonal and
the normalization in the definition of $\var_j$ implies that 
\[
    \vort = [ \var_1 \var_2, \cdots, \var_d]
\]
is a matrix whose nonzero columns form an orthonormal set. Rewriting the objective in terms of
the $(\var_j,y_j)$, we have

\begin{equation}\label{eq:objinY}
\sum\limits_{i=1}^{d} \frac{\encvec_i^{\top} \cov^2 \encvec_i}{\noisevar + \encvec_i^{\top} \cov \encvec_i} = \sum\limits_{i=1}^{d} \var_i^{\top} \cov \var_i \frac{y_i}{\noisevar + y_i} = \sum\limits_{i=1}^{d} \var_i^{\top} \cov \var_i m_i,
\end{equation}	
where  $m_i = \frac{y_i}{\noisevar+y_i}$. 
Observe that under this new parametrization, the constraint only depends on $\{y_i \}_{i=1}^{d}$. Without loss of generality, we assume that $y_1 \geq y_2 \geq \cdots \ge y_d$, implying that $m_1 \geq m_2 \geq \cdots \geq m_d$.
We now prove that for given $\{y_i \}_{i=1}^{d}$, choosing the $\var_i$ along the eigenvectors of $\cov$ is optimal.


Denote the diagonal elements of $\vort^\top \cov \vort$ by $\{\lambda_i^2\}_{i = 1}^d$ and let  $\{ \lambda_{i, \downarrow}^2 \}_{i=1}^{d}$  denote the same diagonal elements arranged in descending order. Denote the eigenvalues of $\vort^{\top} \cov \vort$ by $\{ \orthvar_i^2 \}_{i=1}^{d}$ where $\orthvar_1 \geq \orthvar_2 \geq \cdots \geq \orthvar_d$. 
Again invoking the Schur-Horn theorem, the eigenvalues of $\vort^{\top} \cov \vort$
 majorize its diagonal entries
\begin{equation}\label{eq:majorization}
\left \lbrace \lambda_i^2 \right \rbrace_{i=1}^{d} \prec \left \lbrace \orthvar_i^2 \right \rbrace_{i=1}^{d}.
\end{equation} 
 Substituting $\lambda_i^2 = \var_i^{\top} \cov \var_i$ in \eqref{eq:objinY}, we have
\begin{align*}
\sum\limits_{i=1}^{d} \lambda_i^2 m_i \stackrel{(a)}{\leq} \sum\limits_{i=1}^{d} \lambda_{i,\downarrow}^2 m_i  &= \lambda_{1,\downarrow}^2 m_1 + \sum\limits_{i=2}^{d} \left( \sum\limits_{j=1}^{i} \lambda_{j,\downarrow}^2 -\sum\limits_{j=1}^{i-1} \lambda_{j,\downarrow}^2 \right) m_i \\
&= \lambda_{1,\downarrow}^2 m_1 +  \sum\limits_{i=2}^{d} m_i \sum\limits_{j=1}^{i} \lambda_{j,\downarrow}^2 - \sum\limits_{i=2}^{d} m_i  \sum\limits_{j=1}^{i-1} \lambda_{j,\downarrow}^2 \\
&=  \lambda_{1,\downarrow}^2 (m_1 - m_2) + m_d \left( \sum\limits_{j=1}^{d} \lambda_{j,\downarrow}^2 \right) + \sum\limits_{i=2}^{d-1} \left( m_i - m_{i+1}\right) \sum\limits_{j=1}^{i} \lambda_{j,\downarrow}^2 \\ 
&\stackrel{(b)}{\leq} \orthvar_1^2 (m_1 - m_2) + m_d \left( \sum\limits_{j=1}^{d} \orthvar_j^2 \right) + \sum\limits_{i=2}^{d-1} (m_i - m_{i+1}) \sum\limits_{j=1}^{i} \orthvar_j^2\\
 &\stackrel{(c)}{\leq} \sigma_1^2 (m_1 - m_2) + m_d \left( \sum\limits_{j=1}^{d} \sigma_j^2 \right) + \sum\limits_{i=2}^{d-1} (m_i - m_{i+1}) \sum\limits_{j=1}^{i} \sigma_j^2\\
&= \sum\limits_{i=1}^{d} \sigma_i^2 m_i,
\end{align*}

where inequality $(a)$ follows from the assumption that $m_1 \geq m_2 \geq \cdots \geq m_d$, and $(b)$ from the definition in \eqref{eq:majorization}. Since $\vort $'s nonzero columns form an orthonormal set, the eigenvalues of $\vort^{\top} \cov \vort$, when arranged in descending order, are at most the eigenvalues of $\cov$ from Corollary 4.3.37 in \cite{HornJ85}, and therefore $(c)$ follows.

This upper bound is attained when $\var_i = \eigvec_i$ for nonzero $y_i$, where $\eigvec_i$ is the normalized eigenvector of $\cov$ corresponding to eigenvalue $\sigma_i^2$. To see this, note that when $\var_i = \eigvec_i$, $\lambda_i^2 = \orthvar_i^2 = \sigma_i^2$. From the definition of $\var_i, \encvec_i = \cov^{-1/2}\eigvec_i \sqrt{y_i} = \eigvec_i \frac{\sqrt{y_i}}{\sigma_i}$. Therefore, for a Schur-concave $\rateproxy$, the set of matrices whose nonzero columns are eigenvectors of $\cov$ is optimal. We now prove that for a strictly Schur-concave $\rateproxy$, if $\cov$ has distinct eigenvalues, this set contains all of the optimal solutions $\encnew$.


We know that for a fixed $y_1 \geq y_2 \geq \cdots \geq y_d$, (implying a fixed $m_1 \geq m_2 \geq \cdots \geq m_d$) the upper bound $\sum\limits_{i=1}^{d} \sigma_i^2 m_i$ is attained by the previous choice of $\var_i$. Note that if all nonzero $m_i$ are distinct, equality in $(b)$ and $(c)$ is attained if and only if the nonzero diagonal elements of $\vort^{\top} \cov \vort$ equal the corresponding eigenvalues of $\cov$. This implies that, if all nonzero $m_i$ are distinct, the upper bound is attained if and only if $\var_i = \eigvec_i$ for nonzero $y_i$. Therefore, it is sufficient to prove that for the following optimization problem

\begin{equation} \label{eq:ubschur}
\begin{aligned} 
\sup\limits_{\{y_i \geq 0\}} \quad & \sum\limits_{i=1}^{d} \sigma_i^2 \frac{y_i}{\noisevar + y_i}  \\
\text{subject to} \quad & R \geq \rho\left(\left\{y_i \right\}_{i= 1}^d\right),
\end{aligned}
\end{equation}
any optimal $\{y_i\}$ must be such that the nonzero $y_i$ are distinct. Firstly, note that since $\sigma_1^2 > \sigma_2^2 > \cdots > \sigma_d^2$, we must have $y_1 \geq y_2 \geq \cdots \geq y_d$. Assume to the contrary that for an optimal $\{ y_i \}_{i=1}^{d}$ there exists $1 \leq j, \ell < d$ such that $y_{j-1} > y_j = y_{j+1} = y_{j+2} = \cdots = y_{j+\ell} > y_{j+\ell+1} \geq 0$, where $y_0$ is chosen to be any real number strictly greater than $y_1$ and $y_{d+1}=0$. 
Take $\delta > 0$ small.
Denote a new sequence $ \{ y'_i \}_{i=1}^{d}$ where $y'_{j} = y_j+\delta, y'_{j+\ell} = y_{j+\ell} - \delta$ and $y'_i = y_i$ for $1\leq i \leq d$ with $i \neq j$ and $j+\ell$. 
Since $\rateproxy$ is strictly Schur-concave, the constraint is strictly
improved,
\[ \rateproxy\left(  \{y'_i \}_{i=1}^{d} \right) < \rateproxy\left(  \{y_i \}_{i=1}^{d} \right). \]
Since $\sigma^2_j > \sigma^2_{j + \ell}$,  the objective is strictly
improved for sufficiently small $\delta$,
$$
\sum\limits_{i=1}^{d} \sigma_i^2 \frac{y_i}{\noisevar + y_i} <
\sum\limits_{i=1}^{d} \sigma_i^2 \frac{y'_i}{\noisevar + y'_i},
$$
as desired.
\end{proof}

As a consequence of Theorem~\ref{thm:diagisoptimal}, encoding via an optimal $\encnew$ can be viewed as a projection along the eigenvectors of $\cov$, followed by different scalings applied to each component, i.e. $\encnew  = \eigmat \scalemat$ where $\scalemat$ is a diagonal matrix with entries $\scaling_i \geq 0$ and 
$\eigmat$ is the normalized eigenvector matrix. Only $\scalemat$ remains to be determined, and to this end, we may assume that
$\cov$ is diagonal with nonincreasing diagonal entries, implying $\eigmat = \bm{I}$. In subsequent sections, our choice of $\rateproxy$ will be of the form $\sum\limits_{i=1}^{d} \rateproxy_{sl}$, where $\rateproxy_{sl}: \RR_{\geq 0} \rightarrow \RR_{\geq 0}$\footnote{``sl" stands for single-letter} is (strictly) concave, making $\rateproxy$ (strictly) Schur-concave (see Proposition~\ref{prop:single-letter} in Appendix~\ref{app:schur}). Therefore,  \eqref{eq:MASTEROPTI_MMSE} reduces to

\begin{equation} \label{eq:MASTEROPTI_S}
\begin{aligned} 
\inf\limits_{\scalemat} \quad & \tr(\cov) - \tr(\cov \scalemat (\scalemat^{\top} \cov \scalemat + \noisevar\bm{I} )^{-1}\scalemat^{\top} \cov ) \\
\text{subject to} \quad & R \geq \rateproxy_{sl}\left( \{\scaling_i^2 \sigma^2_i \} \right),
\end{aligned}
\end{equation}
where the infimum is over diagonal matrices $\scalemat$. To handle situations for which 
\begin{equation}
    \lim_{s \rightarrow \infty} \rateproxy_{sl}(s) < \infty,
\end{equation}
we allow the diagonal entries of $\scalemat$ to be $\infty$, with the objective for such cases defined via its continuous extension.

In the next section, we will solve \eqref{eq:MASTEROPTI_S} for several specific choices of $\rateproxy_{sl}$.

\section{Explicit Solutions: Conventional Linear Autoencoders and PBA} \label{sec:PBAPCA}
\subsection{Conventional Linear Autoencoders} \label{subsec:PCA}

Given a centered dataset $\data_1, \data_2, \cdots, \data_n \in \RR^{d}$, consider a linear autoencoder optimization problem where the encoder and decoder, $\encnew$ and $ \decnew$, respectively, are $d$-by-$k$ matrices where $k \leq d$ is a parameter. The goal is to minimize the mean squared error as given by \eqref{eq:1stmse}. PCA corresponds to the global optimal solution of this optimization problem, where $\encnew = \decnew = \eigmat_k$, where $\eigmat_k \in \RR^{d \times k}$ is a matrix whose columns are the $k$ eigenvectors corresponding to the $k$ largest eigenvalues of $\cov$. However, there are multiple global optimal solutions, given by any encoder-decoder pair of the form $\left( \eigmat_k  \bm{V},  \eigmat_k  \bm{V} \right)$, where $\bm{V}$ is an orthogonal matrix \cite{BaldiH89}.

We now recover linear autoencoders through our framework in Section~\ref{sec:FRAMEWORK}. Consider the optimization problem in \eqref{eq:MASTEROPTI_S} where $\rateproxy_{sl}: \RR_{\geq 0} \rightarrow \{0,1 \}$ is a concave function defined as
\begin{equation} \label{eq:PCA_F}
\rateproxy_{sl}(x) = \bm{1} \left[ x > 0 \right].
\end{equation}

Note that this penalizes the dimension of the latents, as desired.
Note also that this cost is Schur-concave but not strictly so.
The fact that PCA solves conventional linear autoencoding,
but is not necessarily the unique solution, follows immediately
from Theorem~\ref{thm:diagisoptimal}.

\begin{theorem}\label{thm:PCA}
    If $\rateproxy_{sl}(\cdot)$ is given by \eqref{eq:PCA_F}, then an optimal solution for \eqref{eq:MASTEROPTI_S} is given by a diagonal matrix $\scalemat$ whose top $\min(\lfloor R \rfloor, d)$ diagonal entries are equal to $\infty$ and the remaining entries are $0$.
\end{theorem}
\begin{proof}
Let $\cF \ed \left \lbrace i \in \left[ \ipdim \right] : \scaling_i > 0 \right \rbrace$, implying $\left \lvert \cF \right \rvert \leq R$. Since $\cov$ and $\scalemat$ are diagonal, the optimization problem in \eqref{eq:MASTEROPTI_S} can be written as 

\begin{equation} \label{eq:PCA_OPT}
\begin{aligned} 
\inf\limits_{\left \lbrace \scaling_{\ell} \right \rbrace} & \sum\limits_{j \in \left[d\right] \backslash \cF} \sigma_j^2 + \sum\limits_{ \ell \in \cF} \frac{\noisevar \sigma_{\ell}^2}{\noisevar+ \sigma_{\ell}^2 \scaling_{\ell}^2} \\
\text{subject to} \quad & R \geq \sum\limits_{i=1}^{d} \bm{1}\left[ \scaling_i > 0 \right].
\end{aligned}
\end{equation}

 Since the value of $\scaling_{\ell}, \ell \in \cF$ does not affect the rate constraint, each of the $\scaling_{\ell}$ can be made as large as possible without changing the rate constraint. Therefore, the infimum value of the objective is $\sum\limits_{j \in \left[ d \right] \backslash \cF} \sigma^2_j$. Since we seek to minimize the distortion, the optimal $\cF$ is the set of indices with the largest $\left \lvert \cF \right \rvert$ eigenvalues. Since the number of these eigenvalues cannot exceed $R$,  we choose $\left \lvert \cF \right \rvert = \min(\lfloor R \rfloor,d)$. 
\end{proof}

Unlike the conventional linear autoencoder framework, in Section~\ref{sec:FRAMEWORK}, the latent variables $\encnew^{\top} \data$ are quantized, which we model with additive white noise of fixed variance.  Therefore, an infinite value of $\scaling_i$ indicates sending $\eigvec_i^{\top} \data$ with full precision where $\eigvec_i$ is the eigenvector corresponding to the $i^{th}$ largest eigenvalue. This implies that PCA with parameter $k$ corresponds to $\encnew = \eigmat \scalemat$, where $\scalemat$ is a diagonal matrix whose top $k$ diagonal entries are equal to $\infty$ and the  $d-k$ remaining diagonal entries are $0$. Therefore, for any $R$ such that $\lfloor R \rfloor = k$, an optimal solution to \eqref{eq:MASTEROPTI_S} corresponds to linearly projecting the data along the top $k$ eigenvectors, which is the same as PCA. Note that, like \cite{BaldiH89}, we only prove that projecting along the eigenvectors is one of possibly other optimal solutions. However, even a slight amount of curvature in $\rateproxy$ would make it strictly Schur-concave, thus recovering the principal directions.
We next turn to a specific cost function with curvature, namely
the PBA cost function that was our original motivation.

\subsection{Principal Bit Analysis (PBA)}

Consider the choice of $\rateproxy_{sl}:\RR_{\geq 0} \rightarrow \RR_{\geq 0}$ that provided the
original impetus for Theorem~\ref{thm:diagisoptimal}. For $\gamma>\frac{2}{\sigma^2}$,

\begin{equation}\label{eq:LOG_F}
\rateproxy_{sl}(x) = \frac{1}{2} \log(\gamma x+1).
\end{equation}

The nature of the optimization problem depends on the value 
of $\gamma$. For $1 \le \gamma \sigma^2 \le 2$, the problem can be
made convex with a simple change of variable. For $\gamma \sigma^2= 1$,
the problem coincides with the classical waterfilling procedure
in rate-distortion theory, in fact. For $\gamma \sigma^2 > 2$, the problem
is significantly more challenging. Since we are interested in
relatively large values of $\gamma$ for our compression application
(see Section~\ref{sec:EXPERIMENTS} to follow), we focus on the case $\gamma > 2/\sigma^2$.

\begin{algorithm}[h]
	\caption{Principal Bit Analysis (PBA)} \label{alg:PBA}
	\begin{algorithmic}[1]
		\Require $\lambda > 0$, $\scalePBA  > 2$, 
		\begin{equation}
            \cov = \left[ \begin{array}{cccc}
                \sigma^2_1 & 0 & \cdots & 0 \\
                0 & \sigma^2_2 & \cdots & 0 \\
                \vdots & \vdots & \ddots & \vdots \\
                0 & 0 & \cdots & \sigma^2_\ipdim \\
            \end{array} \right] \succ 0,
            \end{equation}
		such that $\sigma^2_1 \ge \sigma^2_2 \ge \cdots \ge \sigma^2_\ipdim$.
		\State If $\lambda \ge \sigma^2_1/(4(\scalePBA - 1))$, Output  
			   $\bar{R}_{\text{opt}} = 0, \bar{D}_{\text{opt}} = \sum_{i = 1}^\ipdim \sigma^2_i$.
		\State Set $\bar{\ipdim} = \max\left\{i : \lambda < \sigma_i^2/4(\alpha - 1)\right\}$.
			\State Set $\bar{R}$, $\bar{D}$ to zero arrays of size $2\bar{\ipdim}$.
			\For{$\send \in \left \lbrace 1,2, \cdots \bar{\ipdim} \right \rbrace $}
				\State $\bar{D}(2\send-1) =  \sum\limits_{i=1}^{\send} \frac{\sigma_i^2}{2(\scalePBA -1)} \left( 1 - \sqrt{1 - \frac{4 \lambda (\scalePBA -1)}{\sigma_i^2}}\right) + \sum\limits_{i=\send+1}^{\ipdim} \frac{\sigma_i^2}{\scalePBA } $, 
				\State $\bar{R}(2\send-1)= \sum\limits_{i=1}^{\send} \frac{1}{2} \log \left( \frac{\sigma_i^2}{4\lambda} \right) + \log\left( 1 + \sqrt{1 - \frac{4 \lambda (\scalePBA -1)}{\sigma_i^2}}\right)$.
				\State $\bar{D}(2\send) =\left( \sum\limits_{i=1}^{\send-1} \frac{\sigma_i^2}{2(\scalePBA -1)} \left( 1 - \sqrt{1 - \frac{4 \lambda (\scalePBA -1)}{\sigma_i^2}}\right) + \frac{\sigma_{\send}^2}{2(\scalePBA -1)} \left( 1 + \sqrt{1 - \frac{4 \lambda (\scalePBA -1)}{\sigma_{\send}^2}}\right) + \sum\limits_{i=\send+1}^{\ipdim} \frac{\sigma_i^2}{\scalePBA } \right)$.
				\State $\bar{R}(2\send)=\sum\limits_{i=1}^{\send} \frac{1}{2} \log \left( \frac{\sigma_i^2}{4\lambda} \right) + \sum\limits_{i=1}^{\send-1} \log\left( 1 + \sqrt{1 - \frac{4 \lambda (\scalePBA -1)}{\sigma_i^2}}\right) + \log\left( 1 - \sqrt{1 - \frac{4 \lambda (\scalePBA -1)}{\sigma_{\send}^2}}\right)$.
			\EndFor
			\State $\send^* \leftarrow \argmin_{j \in \left[ 2\bar{\ipdim} \right]} \bar{D}(j) + \lambda \bar{R}(j). $
		\State Output $\bar{R}_{\text{opt}} = \bar{R}(\send^*), \bar{D}_{\text{opt}} = \bar{D}(\send^*)$.
			
	\end{algorithmic}
\end{algorithm} 


\begin{theorem}
If $\rateproxy_{sl}(\cdot)$ is given by \eqref{eq:LOG_F} for $\gamma > \frac{2}{\sigma^2}$, then for any $\lambda > 0$,
	the pair $\bar{R}_{\text{opt}}, \bar{D}_{\text{opt}}$ obtained from the output of 
	  Algorithm~\ref{alg:PBA} satisfies 
	  \begin{equation} 
		  \label{eq:PBA_OPT_LAGRANGE}
			   \bar{D}_{\text{opt}} + \lambda \bar{R}_{\text{opt}} = 
			  \inf\limits_{\scalemat} \quad \tr(\cov) - \tr(\cov \scalemat (\scalemat^{\top} \cov \scalemat + \noisevar\bm{I} )^{-1}\scalemat^{\top} \cov )  
			  + \lambda \sum_{i = 1}^\ipdim \rateproxy_{sl}\left( \{\scaling_i^2 \sigma^2_i \} \right),
	  \end{equation}
\end{theorem}

\begin{proof}
Since $\cov$ and $\scalemat$ are diagonal, the optimization problem in 
\eqref{eq:PBA_OPT_LAGRANGE} can be written as 
\begin{equation} \label{eq:PBA_OPT}
\begin{aligned} 
\inf\limits_{\left \lbrace \scaling_i \right \rbrace} \quad & \sum\limits_{i=1}^{d} \frac{\noisevar \sigma_{i}^{2}}{\noisevar+\scaling_{i}^{2} \sigma_{i}^{2}} 
   + \lambda \cdot \frac{1}{2}\sum\limits_{i=1}^{d} \log\left( 1 + \gamma \scaling_{i}^2 \sigma_{i}^{2} \right).
\end{aligned}
\end{equation}
With the following change of variables $ \scalePBA = \gamma \sigma^2 $, $\scaling_{i} \mapsto \scaling_{i}'^2 = \scalePBA \frac{\scaling_{i}'^2}{\sigma^2}$, we obtain

\begin{equation} \label{eq:PBA_OPT_p1}
\begin{aligned} 
\inf\limits_{\left \lbrace \scaling_i' \right \rbrace} \quad & \sum\limits_{i=1}^{d} \scalePBA \frac{ \sigma_{i}^{2}}{\scalePBA+\scaling_{i}^{'2} \sigma_{i}^{2}}  + \lambda \cdot \frac{1}{2}\sum\limits_{i=1}^{d} \log\left( 1 + \scaling_i'^2 \sigma_{i}^{2} \right).
\end{aligned}
\end{equation}

Ignoring the constant factor in the objective, perform the change of variable $\scaling'_{i} \mapsto \distortion_i = \frac{\sigma_i^2}{\scalePBA  +\scaling_{i}'^2\sigma_{i}^{2}}$
to obtain
\begin{equation} \label{eq:PBA_REDUCED_OPT}
\begin{aligned} 
\inf\limits_{\left \lbrace \distortion_i \right \rbrace} \quad & \sum\limits_{i=1}^{\ipdim} \distortion_i  + \frac{\lambda}{2}\sum\limits_{i=1}^{d} \log\left( \frac{\sigma_i^2}{\distortion_i} - (\scalePBA -1) \right), \\
    \text{subject to} \quad &\distortion_i \leq \frac{\sigma_i^2}{\scalePBA } \hspace{5mm} \text{ for all } i \in \left[ \ipdim \right].
\end{aligned}
\end{equation}

This optimization problem is nonconvex since the function $\log\left( \frac{\sigma_i^2}{\distortion_i} - (\scalePBA -1) \right)$ is convex for $0 \leq \distortion_i \leq \frac{\sigma_i^2}{2(\scalePBA -1)}$ but concave for $\frac{\sigma_i^2}{2(\scalePBA -1)} < \distortion_i \leq \frac{\sigma_i^2}{\scalePBA}$ and the
latter interval is nonempty since $\scalePBA>2$. 

	Any optimizing $\{\distortion_i\}$ must be a stationary point of 
\begin{equation}\label{eq:LAG_PBA}
\lag\left( \setD, \lambda, \setmu \right) = \sum\limits_{i=1}^{\ipdim} \distortion_i + \lambda\left( \sum\limits_{i=1}^{\ipdim} \log\left( \frac{\sigma_i^2}{\distortion_i} - (\scalePBA -1) \right) \right) + \sum\limits_{i=1}^{\ipdim} \mu_i \left( \distortion_i - \frac{\sigma_i^2}{\scalePBA } \right). 
\end{equation}
for some $\{\mu_i\}_{i = 1}^\ipdim$ with $\mu_i \geq 0$ for all $i \in \left[ \ipdim \right]$ and satisfying the complementary slackness condition~\cite[Prop.~3.3.1]{Bertsekas:Nonlinear}.  The stationary points satisfy, for each $i$
\begin{equation}\label{eq:LAG_STAT}
\frac{\partial \lag}{\partial \distortion_i} = 1 - \lambda\left( \frac{\frac{\sigma_i^2}{\distortion_i^2}}{\left( \frac{\sigma_i^2}{\distortion_i} - (\scalePBA -1) \right)}\right) + \mu_i = 0.
\end{equation}

	Let $\cF = \left \lbrace i : \distortion_i < \frac{\sigma_i^2}{\scalePBA } \right \rbrace$. 
 For $i \in \cF$, $\mu_i = 0$ due to complementary slackness. Substituting in \eqref{eq:LAG_STAT} we obtain a quadratic equation in $\distortion_i$
 
 \[
 (\scalePBA-1)\distortion_i^2 - \sigma_i^2 \distortion_i + \lambda \sigma_i^2 = 0. 
 \]
 which gives 
 $$ \distortion_i = \frac{\sigma_i^2}{2(\scalePBA -1)} \left( 1 \pm \sqrt{1 - \frac{4\lambda (\scalePBA -1)}{\sigma_i^2}} \right).$$

Let $c_i = \sqrt{1 - \frac{4\lambda (\scalePBA -1)}{\sigma_i^2}}$. Note that $\frac{\sigma_i^2}{2(\scalePBA -1)} \left( 1 + c_i \right)$ is always in the concave region and $\frac{\sigma_i^2}{2(\scalePBA -1)} \left( 1 - c_i \right)$ is always in the convex region for a $\lambda$ chosen such that $\distortion_i$ is a real number strictly less than $\frac{\sigma_i^2}{\scalePBA }$. 


Therefore the optimal set of distortions are contained in the following set of $3^{\ipdim}$ points 
$$ \prod_{i = 1}^\ipdim \left\{   \frac{\sigma_i^2}{2(\scalePBA -1)} \left( 1 + \sqrt{1 - \frac{4\lambda (\scalePBA -1)}{\sigma_i^2}} \right),  \frac{\sigma_i^2}{2(\scalePBA -1)} \left( 1 - \sqrt{1 - \frac{4\lambda (\scalePBA -1)}{\sigma_i^2}} \right), \frac{\sigma_i^2}{\scalePBA}\right\}.$$

We now reduce the size of the above set by making a two observations:

%

\noindent \textbf{(1). $\cF$ is contiguous.}

\begin{lemma}\label{lem:NORM_INC}
There exists an optimal $\left \lbrace \distortion_i^* \right \rbrace_{i=1}^{d}$ for \eqref{eq:PBA_REDUCED_OPT} such that (a) $\frac{\sigma_i^2}{\distortion_i^*}$ is a nonincreasing sequence and (b) $\cF = \{ 1,2,\cdots \left \lvert \cF \right \rvert \}$. 
\end{lemma}

\begin{proof}
Substitute $x_i = \frac{\sigma_i^2}{\distortion_i}$ in \eqref{eq:PBA_REDUCED_OPT}. This gives us 
\begin{equation} \label{eq:PBA_NORM_OPT}
\begin{aligned} 
\inf\limits_{\left \lbrace x_i \right \rbrace} \quad & \sum\limits_{i=1}^{\ipdim} \frac{\sigma_i^2}{x_i} + \frac{\lambda}{2}\sum\limits_{i=1}^{d} \log\left( x_i - (\scalePBA -1) \right), \\
    \text{subject to} \quad & x_i \geq \scalePBA \hspace{5mm} \text{ for all } i \in \left[ \ipdim \right].
\end{aligned}
\end{equation}
Let $\{ x_i^* \}_{i=1}^{d}$ be an optimal solution for \eqref{eq:PBA_NORM_OPT}. If, for $i>j$, $x_i^* > x_j^* \ge \scalePBA$, then exchanging the values provides a solution that has the same rate and lower distortion since $\frac{\sigma_i^2}{x_i^*} + \frac{\sigma_j^2}{x_j^*} \geq \frac{\sigma_i^2}{x_j^*} + \frac{\sigma_j^2}{x_i^*}$. This proves (a). Part (b) follows immediately.
%
\end{proof}

\noindent\textbf{(2). No two solutions are concave.}
\begin{lemma}\label{lem:SINGLE_CONCAVE}
For $R>0$, let $\left \lbrace \distortion_i^* \right \rbrace_{i=1}^{\ipdim}$ be an optimal solution for \eqref{eq:PBA_REDUCED_OPT}. There exists at most one $\distortion_i^*$ such that $\frac{\sigma_i^2}{2(\scalePBA -1)} < \distortion_i^* < \frac{\sigma_i^2}{\scalePBA }$.
\end{lemma}

\begin{proof}
Let $\distortion_i^*, \distortion_j^*$ be such that $\frac{\sigma_i^2}{2(\scalePBA -1)} < \distortion_i^* < \frac{\sigma_i^2}{\scalePBA }$ and $\frac{\sigma_j^2}{2(\scalePBA -1)} < \distortion_j^* < \frac{\sigma_j^2}{\scalePBA }$ . Without loss of generality, assume $\distortion_i^* < \distortion_j^*$. Denote the individual rate constraint function by $r\left( \distortion_i \right) \triangleq \log \left( \frac{\sigma_i^2}{\distortion_i} - (\scalePBA -1) \right)$. Since $r$ is concave in $\left(\frac{\sigma_i^2}{2(a-1)}, \frac{\sigma_i^2}{\scalePBA } \right)$, there exist an $\eps > 0$ such that

\begin{align}
r\left(\distortion_i^* - \eps\right) +r\left( \distortion_j^* + \eps \right) &=  r\left( \distortion_i^* \right) - \eps r'\left( \distortion_i^* \right) + O(\eps^2) + r\left(\distortion_j^*\right) + \eps r'\left(\distortion_j^* \right) + O(\eps^2) \\
 &< r\left( \distortion_i^* \right) + r\left(\distortion_j^*\right)
\end{align}
The last inequality follows from concavity of $r$. Therefore,  replacing $\left(\distortion_i^*, \distortion_j^*\right)$ with $\left(\distortion_i^* - \eps, \distortion_j^*+ \eps \right)$, the rate constraint can be improved while keeping the objective in \eqref{eq:PBA_REDUCED_OPT} constant, contradicting the optimality assumption of $\left \lbrace \distortion_i^* \right \rbrace$.
\end{proof}

There is at most one $\distortion_i^*$ such that $\distortion_i^* = \frac{\sigma_i^2}{2(\scalePBA -1)}\left( 1 + c_i\right)$. Assuming such an $i$ exists, $x_i = \frac{2(\scalePBA -1)}{1+c_i} < 2(\scalePBA -1)$. For the convex roots, $x_i = \frac{2(\scalePBA -1)}{1-c_i} > 2(a-1)$. Therefore from Lemma~\ref{lem:NORM_INC}, all the convex roots are contiguous. Therefore, the set of potentially optimal solutions reduces to cardinality $2d$, where each solution is characterized by the number of components that send non-zero rate and whether or not a concave root is sent. PBA, detailed in Algorithm~\ref{alg:PBA} finds the minimum value of the Lagrangian across these $2d$ solutions for a fixed $\lambda$. 
\end{proof}

Note that by sweeping $\lambda > 0$, one can compute the lower convex envelope of the $(D,R)$ curve.
Since every Pareto optimal $(D,R)$ must be a stationary point of (\ref{eq:PBA_OPT_LAGRANGE}),
one can also use Algorithm~\ref{alg:PBA} to compute the $(D,R)$ curve itself by sweeping $\lambda$
and retaining all those stationary points that are not Pareto dominated.

\section{Application to Variable-Rate Compression}\label{sec:VARIABLE}

We have seen that an autoencoder formulation inspired by 
data compression succeeds in providing guaranteed
recovery the principal source components. Conversely, 
a number of successful multimedia compressors have recently
been proposed that are either related to, or directly
inspired by, autoencoders~\cite{TschannenAL18, TodericiVJHMSC17,
BalleLS16, TodericiOHVMBCS16, TheisSCH17, RippelB17, HabibianRTC19,
AgustssonMTCTBG17, BalleMSHJ18, ZhouCGSW18, AgustssonTMTG19, BalleCMSJAHT20}.
In particular, Ball\'{e} \emph{et al.}~\cite{BalleMSHJ18}
show that the objective minimized by their compressor coincides with 
that of variational autoencoders. Following 
\cite{BalleCMSJAHT20}, we refer to this objective as \emph{nonlinear
transform coding (NTC)}. We next use Theorem~\ref{thm:diagisoptimal}
to show that
any minimizer of the NTC objective is guaranteed to recover
the principal source components if (1) the source is Gaussian,
 (2) the transforms are restricted to be linear,
and (3) the entropy model is \emph{factorized}, as explained below.

Let $\data \sim \cN\left(0,\cov\right)$, where $\cov$ is a positive semidefinite covariance matrix. As before, we consider an autoencoder defined by its encoder-decoder pair $(f,g)$, where for $k\leq d$, $f: \RR^{d} \rightarrow \RR^{k}$ and $g:\RR^{k} \rightarrow \RR^{d}$ are chosen from prespecified classes $\cC_f$ and $\cC_g$. The NTC framework
assumes dithered quantization during training, as in Section~\ref{sec:FRAMEWORK} and \cite{AgustssonT20, ChoiEL19}, and seeks to minimize the Lagrangian
\begin{equation}
\inf_{f \in \cC_f,g \in \cC_g} \mathbb{E}_{\data, \noise} \left[ \left \lVert \data - g\left( Q\left( f(\data) + \noise \right) - \noise \right) \right \rVert_2^2 \right] + \lambda H\left(Q \left( f( \data) + \noise\right) - \noise | \noise \right). 
\end{equation}
where $\lambda > 0$ and $\noise$ has i.i.d.\  $\text{Unif}\left[-0.5,0.5\right]$ components. 
NTC assumes variable-length compression, and the quantity
$$
H\left(Q \left( f( \data) + \noise\right) - \noise | \noise \right)
$$
is an accurate estimate of minimum expected codelength length for the discrete
random vector $Q \left( f( \data) + \noise\right)$.
As we noted in Section~\ref{sec:FRAMEWORK}, \cite{ZamirF92} showed that for any random variable $\data$, $Q\left( \data+ \noise \right) - \noise$ and $\data+\noise$ have the same joint distribution with $\data$. They also showed that $H\left(Q \left( \data + \noise\right) - \noise | \noise \right) = I\left( \data + \noise;\data \right) = h(\data + \noise)$, where
$h(\cdot)$ denotes differential entropy. Therefore, the objective can be written as
\begin{equation} \label{eq:general_lag}
 \inf_{f \in \cC_f,g \in \cC_g} \mathbb{E}_{\data, \noise} \left[ \left \lVert \data - g\left( f\left( \data \right) + \noise \right) \right \rVert_2^2 \right]  + \lambda h\left( f\left( \data \right) + \noise\right).
 \end{equation}
 (Compare eq.(13) in \cite{BalleCMSJAHT20}).

We consider the case in which $\cC_f, \cC_g$ are the class of linear functions. Let $\encnew, \decnew$ be $d$-by-$d$ matrices.  Define $f\left(\data\right) = \encnew^{\top} \data $, $g\left( \data \right) = \decnew \data$. Substituting this in the above equation, we obtain
\begin{equation}\label{eq:RD_LAG}
\inf\limits_{\encnew, \decnew} \mathbb{E}_{\data, \noise} \left[ \left \lVert \data - \decnew \left( \encnew^\top \data + \noise\right) \right \rVert_2^2 \right] \\ + \lambda h\left( \encnew^\top \data + \noise\right).
\end{equation}
Since $\decnew$ does not appear in the rate constraint, the optimal $\decnew$ can be chosen to be
the minimum mean squared error estimator of $\data \sim \cN\left( 0,\cov \right)$ given 
$\encnew^\top \data + \noise$, as in Section~\ref{sec:FRAMEWORK}. This gives
\begin{equation}\label{eq:RD_LAG_MMSE}
	\inf\limits_{\encnew} \tr(\cov) - \tr(\cov \encnew \left(\encnew^{\top} \cov \encnew + \frac{\ide}{12} \right)^{-1}\encnew^{\top} \cov ) + \lambda h\left( \encnew^\top \data + \noise\right).
\end{equation}

As noted earlier, the rate term $h\left( \encnew^\top \data + \noise\right)$ is an accurate 
estimate for the minimum expected length of the compressed representation of $Q \left( \encnew^\top \data+ 
\noise\right)$.
This assumes that the different components of this vector are encoded jointly, however. In
practice, one often encodes them separately, relying on the transform $\encnew$ to 
eliminate redundancy among the components. Accordingly, we replace the rate term with
$$
\sum_{i = 1}^d
h\left( \encvec_i^\top \data + [\noise]_i\right),
$$
to arrive at the optimization problem
\begin{equation}\label{eq:RD_LAG_MMSE_FACTOR}
	\inf\limits_{\encnew} \tr(\cov) - \tr(\cov \encnew \left(\encnew^{\top} \cov \encnew + \frac{\ide}{12} \right)^{-1}\encnew^{\top} \cov )  + 
\lambda \cdot \sum_{i = 1}^d
h\left( \encvec_i^\top \data + [\noise]_i\right).
\end{equation}

\begin{theorem}
Suppose $\cov$ has distinct eigenvalues. 
	Then any $\encnew$ that achieves the infimum in~(\ref{eq:RD_LAG_MMSE_FACTOR}) 
 has the property that all of its nonzero rows are eigenvectors of $\cov$.
\label{thm:variable_rate}
\end{theorem}

\begin{proof}
Since the distribution of $\noise$ is fixed, by the Gaussian assumption on $\data$, 
$ h\left( \encvec_j^\top \data + [\noise]_j\right)$ only depends on
	$\encvec_j$ through $\encvec^{\top}_j \cov \encvec_j$. Thus we may write
	\begin{equation}
		h(\encvec_j^{\top} \data + \epsilon)
		    =  \rateproxy_{sl}(\encvec_j^{\top} \cov \encvec_j).
	\end{equation}
	By Theorem~\ref{thm:diagisoptimal}, it suffices to show that $\rateproxy_{sl}(\cdot)$
	is strictly concave.
	Let $Z$ be a standard Normal random variable and let $\epsilon$
      be uniformly distributed over $[-1/2,1/2]$, independent of $Z$. 
         Then we have
	\begin{align}
		\rateproxy_{sl}(s) 
		 & = h(\sqrt{s} \cdot Z + \epsilon).
	\end{align}
	Thus by de Bruijn's identity~\cite{CoverT06},
	\begin{equation}
		\rateproxy'_{sl}(s) = \frac{1}{2} J(\epsilon + \sqrt{s} \cdot Z),
	\end{equation}
	where $J(\cdot)$ is the Fisher information. To show that $\rateproxy'_{sl}(\cdot)$
	is strictly concave, it suffices to show that $J(\epsilon + \sqrt{s} \cdot Z)$ is strictly
	decreasing in $s$.\footnote{If $g'(\cdot)$ is strictly decreasing then for all $t > s$,
	$g(t) = g(s) + \int_s^t g'(u) du < g(s) + g'(s)(t-s)$ and likewise for $t < s$. 
	 That $g(\cdot)$ is strictly
	concave then follows from the standard first-order test for 
	concavity~\cite{BoydVandenberghe}.}
	To this end, let $t > s > 0$ and let $Z_1$ and $Z_2$ 
	  be i.i.d.\ standard Normal random variables, independent of $\epsilon$. Then
	  \begin{equation}
		   J(\epsilon + \sqrt{t} \cdot Z) = J(\epsilon + \sqrt{s}\cdot Z_1 + \sqrt{t-s}\cdot Z_2)
	  \end{equation}
	  and by the convolution inequality for Fisher information~\cite{Blachman65},
	  \begin{align}
\frac{1}{J(\epsilon + \sqrt{s} \cdot Z_1 + \sqrt{t - s} \cdot Z_2)}  >
		       \frac{1}{J(\epsilon + \sqrt{s} \cdot Z_1)} + \frac{1}{J(\sqrt{t-s} \cdot Z_2)}> \frac{1}{J(\epsilon + \sqrt{s} \cdot Z_1)},
	  \end{align}
      where the first inequality is strict because $\epsilon + \sqrt{s} \cdot Z_1$
       is not Gaussian distributed.
\end{proof}

\section{Compression Experiments}\label{sec:EXPERIMENTS}

We validate the PBA algorithm experimentally by comparing the performance
of a PBA-derived fixed-rate compressor against  the performance of baseline fixed-rate
compressors. The code of our implementation can be found at \url{https://github.com/SourbhBh/PBA}. As we noted in the previous section, although variable-rate codes are more commonplace in practice, fixed-rate codes do offer some advantages over their more
general counterparts:
\begin{enumerate}
    \item In applications where a train of source
        realizations are compressed sequentially, fixed-rated 
        coding allows for simple concatenation of the compressed
        representations. Maintaining synchrony between the encoder
        and decoder is simpler than with variable-rate codes.
    \item In applications where a dataset of source realizations
        are individually compressed, fixed-rate coding allows for random access
        of data points from the compressed representation.
    \item In streaming in which a sequence of realizations will 
           be streamed, bandwidth provisioning is 
        simplified when the bit-rate is constant over time.
\end{enumerate}

Fixed-rate compressors exist for specialized sources such as
speech~\cite{McCreeB95, SchroederA85} and audio more 
generally~\cite{Vorbis}. We
consider a general-purpose, learned, fixed-rate compressor
derived from PBA and the following two quantization operations.
The first, $Q_{CD}(a,\sigma^2,U,x)$\footnote{``CD'' stands for 
``clamped dithered.''}
accepts the hyperparameter $a$, a variance estimate $\sigma^2$, a
dither realization $U$, and the scalar source realization to be
compressed, $x$, and outputs (a binary representation of) the
nearest point to $x$ in the set
\begin{equation}
    \label{eq:quantizeset}
    \left\{i + U : i \in \mathbb{Z} \ \text{and} \ i + U 
         \in \left(-\frac{\Gamma}{2},\frac{\Gamma}{2}\right] \right\}\newest{,}
\end{equation}
where
\begin{equation}
    \Gamma = 2^{\lfloor \frac{1}{2} \log_2 (4a^2 \sigma^2 + 1) \rfloor }.
\end{equation}
This evidently requires $\log_2 \Gamma$ bits.  The second
function, $Q_{CD}'(a^2,\sigma^2,U,b)$, where $b$ is a binary
string of length $\log_2 \Gamma$, maps the binary
representation $b$ to the point in~(\ref{eq:quantizeset}).
These quantization routines are applied separately to each
latent component. The $\sigma^2$ parameters are
determined during training. The dither $U$ is chosen
uniformly over the set $[-1/2,1/2]$, independently for
each component. We assume that $U$ is chosen pseudorandomly
from a fixed seed that is known to both the encoder and
the decoder. As such, it does not need to be explicitly 
communicated. For our experiments, we fix the $a$ parameter
at $15$ and hard code this both at the encoder and at
the decoder. We found that this choice balances the 
dual goals of minimizing the excess distortion due to 
the clamping quantized points to the interval $(\Gamma/2,\Gamma/2]$ 
and minimizing the rate.

PBA compression proceeds by applying Algorithm~\ref{alg:PBA}
to a training set to determine the matrices $\encnew$ and 
$\decnew$. The variance estimates $\sigma_1^2,\ldots,\sigma_d^2$
for the $d$ latent variances are chosen as the empirical 
variances on the training set and are hard-coded in the encoder
and decoder. Given a data point $\data$,
the encoded representation is the concatenation of the bit strings
$b_1, \ldots, b_d$, where
$$
b_i = Q_{CD}(a^2,\sigma^2_i,U_i,\encvec_i^\top x), 
$$ 
The decoder parses the received bits into $b_1, \ldots, b_d$. 
and computes the latent reconstruction $\hat\latent$, where
$$
\hat\latent_i = Q_{CD}'(a^2,\sigma^2_i,U_i,b_i), 
$$
The reconstruction is then $\decnew\hat\latent$. 

We evaluate the PBA compressor on MNIST \cite{LecunBBH98}, CIFAR-10 \cite{Krizhevsky09}, MIT Faces Dataset, Free Spoken Digit Dataset (FSDD) \cite{Zohar} and a synthetic Gaussian dataset. The synthetic Gaussian dataset is generated from a diagonal covariance matrix obtained from the eigenvalues of the Faces Dataset. 
We compare our algorithms primarily using mean-squared error since our theoretical analysis uses mean squared error as the distortion metric. Our plots display Signal-to-Noise ratios (SNRs) for ease of interpretation. For image datasets, we also compare our algorithms using the Structural Similarity (SSIM) or the Multi-scale Strctural Similarity (MS-SSIM) metrics when applicable \cite{WangBSS04}. We also consider errors on downstream tasks, specifically classification, as a distortion measure.

For all datasets, we compare the performance of the PBA compressor 
against baseline scheme derived from PCA that uses $Q_{CD}$
and $Q_{CD}'$.
The PCA-based scheme sends some of the principal components 
essentially losslessly, and no information about the others.
Specifically, in the context of our framework,
for any given $k$, we choose the first $k$ columns
of $\encnew$ to be aligned with the first $k$ principal components
of the dataset; the remaining columns are zero. Each nonzero column
is scaled such that its Euclidean length multiplied by the eigenvalue has all the significant digits. This is done so that at high rates, the quantization procedure sends the $k$ principal components losslessly. The quantization
and decoder operations are as in the PBA-based scheme; in particular
the $a^2$ parameter is as specified above. By varying $k$, we trade off rate and distortion. 





\subsection{SNR Performance}

\begin{figure}[hbt]
\centering
\minipage{0.18\textwidth}
  \includegraphics[width=\linewidth]{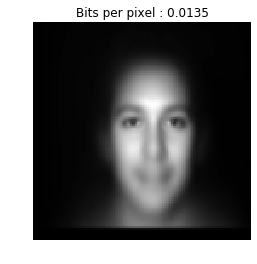}
\endminipage\hfill
\minipage{0.18\textwidth}
  \includegraphics[width=\linewidth]{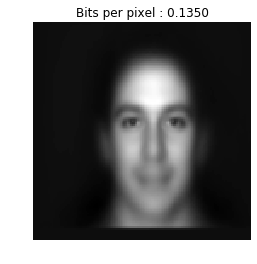}
\endminipage\hfill
\minipage{0.18\textwidth}
  \includegraphics[width=\linewidth]{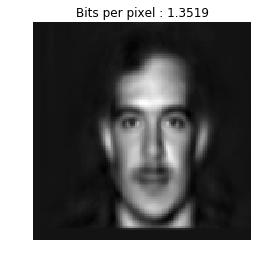}
\endminipage\hfill
\minipage{0.18\textwidth}
  \includegraphics[width=\linewidth]{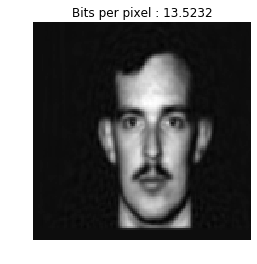}
\endminipage\hfill
\minipage{0.18\textwidth}
  \includegraphics[width=\linewidth]{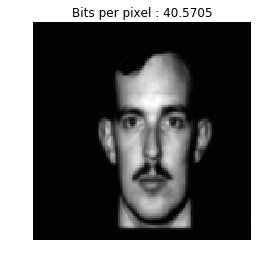}
\endminipage\hfill 
\minipage{0.18\textwidth}
  \includegraphics[width=\linewidth]{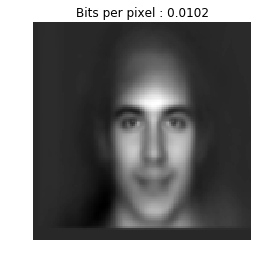}
\endminipage\hfill
\minipage{0.18\textwidth}
  \includegraphics[width=\linewidth]{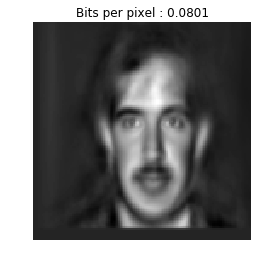}
\endminipage\hfill
\minipage{0.18\textwidth}
  \includegraphics[width=\linewidth]{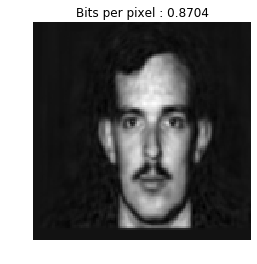}
\endminipage\hfill
\minipage{0.18\textwidth}
  \includegraphics[width=\linewidth]{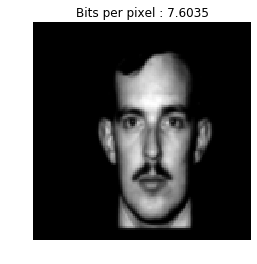}
\endminipage\hfill
\minipage{0.18\textwidth}
  \includegraphics[width=\linewidth]{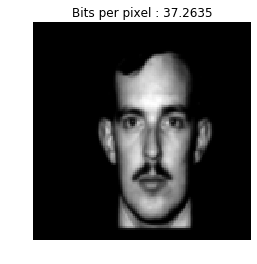}
\endminipage\hfill
\caption{Reconstructions at different bits/pixel values for PCA (top) and PBA (bottom)}\label{fig:reconstructions_pca_pba}
\end{figure}

\begin{figure}[hbt]
\centering
\minipage{0.24\textwidth}
  \includegraphics[width=\linewidth]{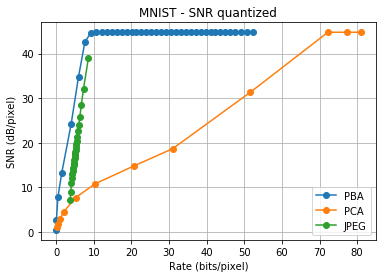}
\endminipage\hfill
\minipage{0.24\textwidth}
  \includegraphics[width=\linewidth]{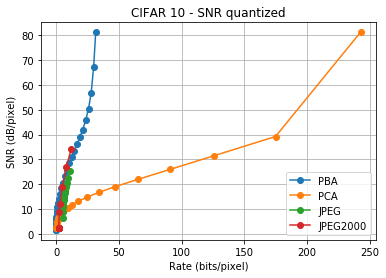}
\endminipage\hfill
\minipage{0.24\textwidth}%
  \includegraphics[width=\linewidth]{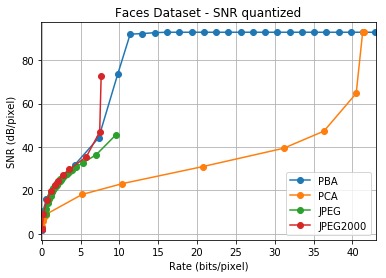}
\endminipage \hfill
\minipage{0.24\textwidth}%
  \includegraphics[width=\linewidth]{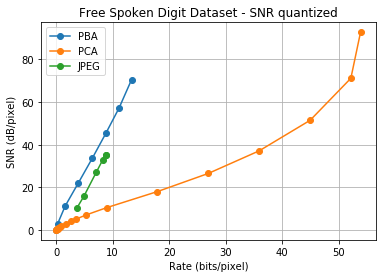}
\endminipage \hfill
\minipage{0.24\textwidth}
  \includegraphics[width=\linewidth]{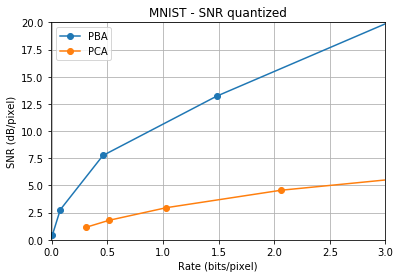}
\endminipage\hfill
\minipage{0.24\textwidth}
  \includegraphics[width=\linewidth]{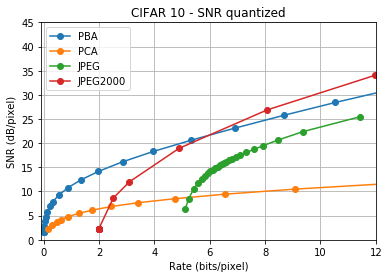}
\endminipage\hfill
\minipage{0.24\textwidth}%
  \includegraphics[width=\linewidth]{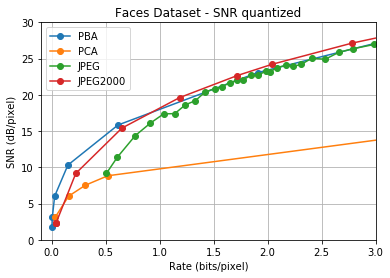}
  \endminipage \hfill
  \minipage{0.24\textwidth}%
  \includegraphics[width=\linewidth]{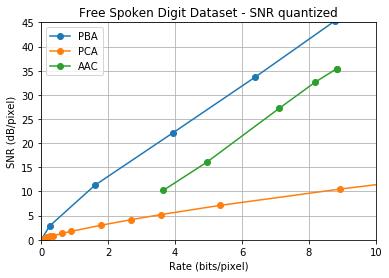}

\endminipage
\caption{SNR/pixel vs Rate (bits/pixel) for MNIST, CIFAR-10, Faces, FSDD datasets. Figures in the bottom row are zoomed-in. } \label{fig:snr}
   
\end{figure}

\begin{figure}[hbt]
\centering
\minipage{0.24\textwidth}
  \includegraphics[width=\linewidth]{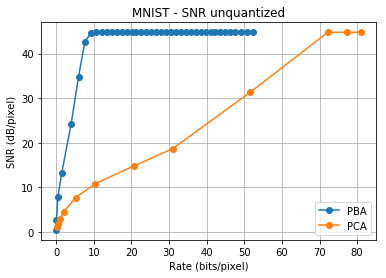}
\endminipage\hfill
\minipage{0.24\textwidth}
  \includegraphics[width=\linewidth]{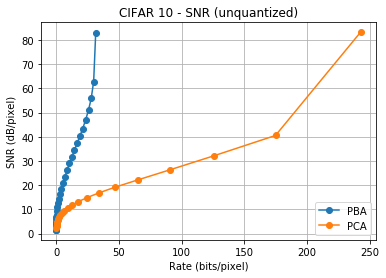}
\endminipage\hfill
\minipage{0.24\textwidth}%
  \includegraphics[width=\linewidth]{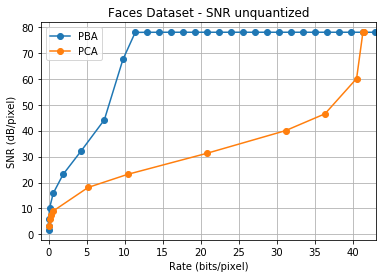}
  
\endminipage
\minipage{0.24\textwidth}%
  \includegraphics[width=\linewidth]{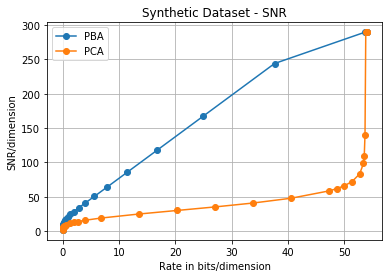}
  
\endminipage

\minipage{0.24\textwidth}
  \includegraphics[width=\linewidth]{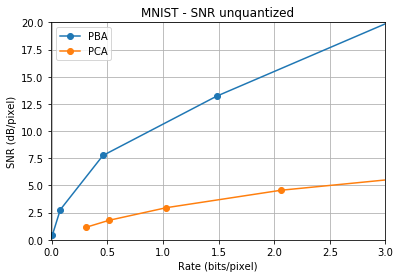}
\endminipage\hfill
\minipage{0.24\textwidth}
  \includegraphics[width=\linewidth]{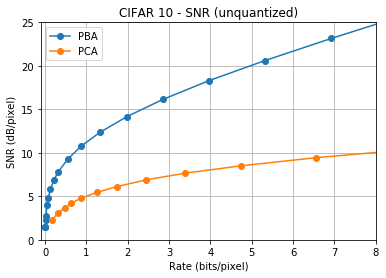}
\endminipage\hfill
\minipage{0.24\textwidth}%
  \includegraphics[width=\linewidth]{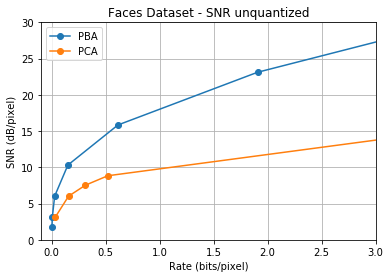}
  
\endminipage
\minipage{0.24\textwidth}%
  \includegraphics[width=\linewidth]{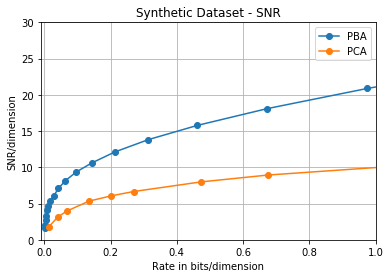}
  
\endminipage

\caption{SNR/pixel vs Rate (bits/pixel) for MNIST, CIFAR-10, Faces and Synthetic dataset. Reconstructions are not rounded to integers from $0$ to $255$. The bottom four plots are zoomed-in versions of the top four plots.}
\label{fig:nq}
\end{figure}

\begin{figure}[hbt]
\centering
\minipage{0.24\textwidth}
  \includegraphics[width=\linewidth]{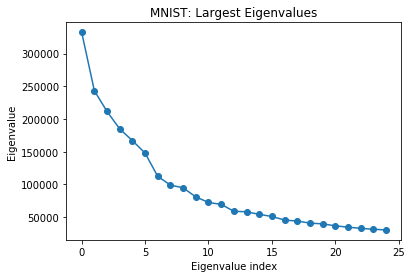}
\endminipage\hfill
\minipage{0.24\textwidth}
  \includegraphics[width=\linewidth]{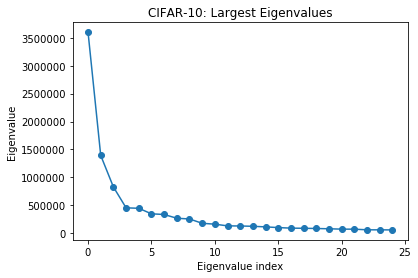}
\endminipage\hfill
\minipage{0.24\textwidth}%
  \includegraphics[width=\linewidth]{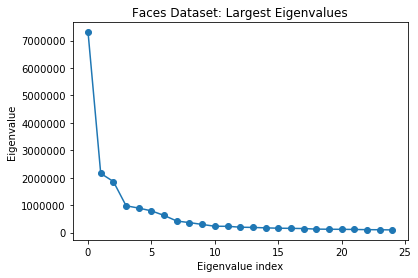}
\endminipage\hfill
\minipage{0.24\textwidth}%
  \includegraphics[width=\linewidth]{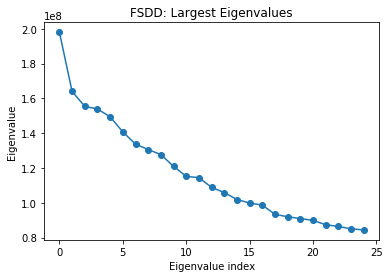}
  \endminipage

\minipage{0.24\textwidth}
  \includegraphics[width=\linewidth]{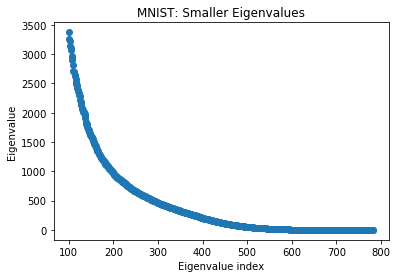}
\endminipage\hfill
\minipage{0.24\textwidth}
  \includegraphics[width=\linewidth]{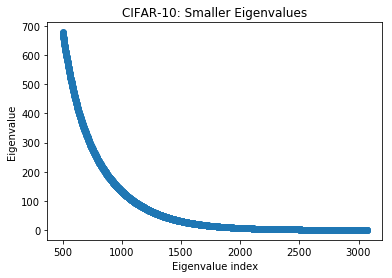}
\endminipage\hfill
\minipage{0.24\textwidth}%
  \includegraphics[width=\linewidth]{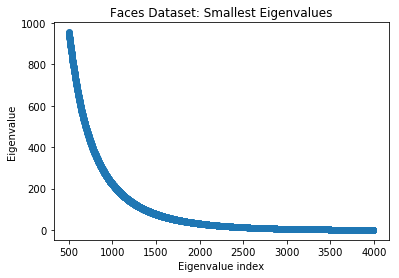}
\endminipage\hfill
\minipage{0.24\textwidth}%
  \includegraphics[width=\linewidth]{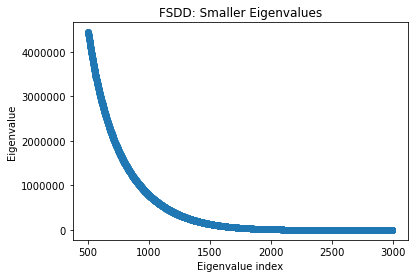}
  \endminipage
\caption{Eigenvalue distribution of the datasets. The top three plots are the largest 25 eigenvalues for MNIST, CIFAR-10, Faces and FSDD dataset. The bottom four figures plot the remaining eigenvalues except the largest 500. }
\label{fig:eig}
\end{figure}

\begin{figure}[hbt]
\centering
\minipage{0.32\textwidth}
  \includegraphics[width=\linewidth]{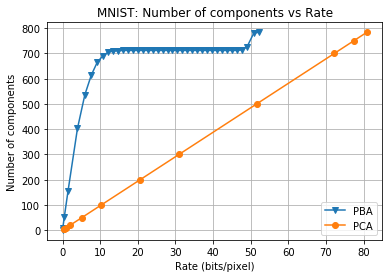}
\endminipage\hfill
\minipage{0.32\textwidth}
  \includegraphics[width=\linewidth]{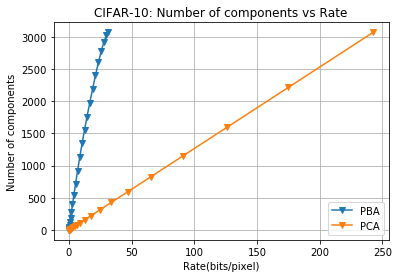}
\endminipage\hfill
\minipage{0.32\textwidth}%
  \includegraphics[width=\linewidth]{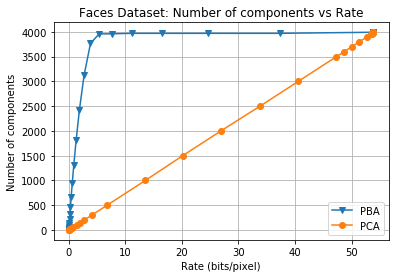}
\endminipage
\caption{Plots of number of components sent vs rate (bits/pixel) for PBA and PCA. }\label{fig:num_comp}
\end{figure}

We begin by examining compression performance under
mean squared error, or equivalently, the SNR, defined as
$$
\text{SNR} = 10 \cdot \log_{10}\left( 
                \frac{P}{\text{MSE}}\right).
$$
where $P$ is the empirical second moment of the dataset.
This was the objective that PBA (and PCA) is designed to minimize.


In Figure~\ref{fig:reconstructions_pca_pba}, we display reconstructions 
for a particular image in the Faces Dataset under PBA and PCA.
Figure~\ref{fig:snr} shows the tradeoff for PBA and PCA against
JPEG and JPEG2000 (for the image datasets) and 
AAC (for the audio dataset). All of the image datasets
have integer pixel values between 0 and 255. Accordingly, we round
the reconstuctions of PBA and PCA to the nearest integer in this range. Figure~\ref{fig:nq} shows the same tradeoff for PBA and PCA when reconstructions are not rounded off to the nearest integer.
We see that PBA consistently
outperforms PCA and JPEG, and is competitive with JPEG2000, even 
though the JPEG and JPEG2000 are variable-rate.  \footnote{It should be 
noted, however, that JPEG and JPEG2000 aim to minimize subjective
distortion, not MSE, and they 
do not allow for training on sample images, as PBA and PCA do. A
similar caveat applies to AAC.} We estimate the size of the JPEG header by compressing an empty image and subtract this estimate from all the compression sizes produced by JPEG. We do not plot JPEG2000 performance for MNIST since it requires at least a 32x32 image. For audio data, we observe that PBA consistently outperforms PCA and AAC.
Since the image data all use $8$ bits per pixel, one can 
obtain infinite SNR at this rate via the trivial encoding that
communicates the raw bits. PCA and PBA do not find this solution
because they quantize in the transform domain, where the
lattice-nature of the pixel distribution is not apparent.
Determining how to leverage lattice structure in the source
distribution for purposes of compression is an interesting
question that transcends the PBA and PCA algorithms and that
we will not pursue here.

The reason that PCA performs poorly is that it 
favors sending the less significant bits of the most significant 
components over the most significant bits of less significant 
components, when the latter are more valuable for reconstructing
the source.  Arguably, it does not identify the ``principal bits.''
Figure~\ref{fig:eig} shows the eigenvalue distribution of the 
different datasets, and Figure~\ref{fig:num_comp} shows the number of distinct
components about which information is sent as a function of rate
for both PBA and PCA. We see that PBA sends information about many 
more components for a given rate than does PCA. We discuss the
ramifications of this for downstream tasks, such as classification,
in Section~\ref{subsec:downstream}.

\subsection{SSIM Performance}

Structural similarity (SSIM) and Multi-Scale Structural similarity (MS-SSIM) are metrics that are tuned to perceptual similarity. Given two images, the SSIM metric outputs a real value between $0$ and $1$ where a higher value indicates more similarity between the images. We evaluate the performance of our algorithms on these metrics as well in Figure~\ref{fig:ssim}. We see that PBA
consistently dominates PCA, and although it was not optimized for
this metric, beats JPEG at low rates as well.

\begin{figure}[hbt]
\centering
\minipage{0.32\textwidth}
  \includegraphics[width=\linewidth]{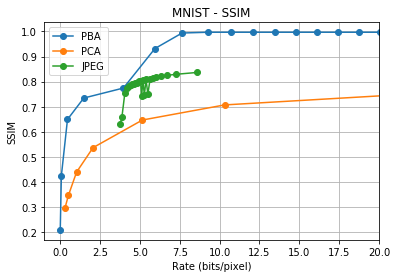}
\endminipage\hfill
\minipage{0.32\textwidth}
  \includegraphics[width=\linewidth]{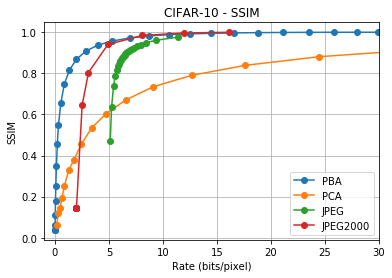}
\endminipage\hfill
\minipage{0.32\textwidth}%
  \includegraphics[width=\linewidth]{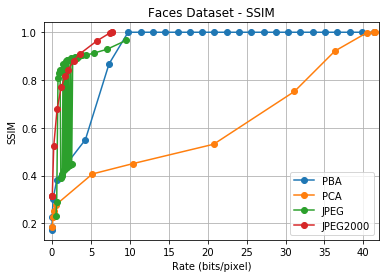}
\endminipage
\caption{SSIM vs Rate (bits/pixel) for MNIST, CIFAR-10, Faces Dataset}\label{fig:ssim}
\end{figure}

\subsection{Performance on Downstream tasks}
\label{subsec:downstream}
Lastly, we compare the impact of using PBA and PCA on an important downstream task, namely classification. We evaluate the algorithms on MNIST and CIFAR-10 datasets and use neural networks for classification. Our hyperparameter and architecture choices are given in Table~\ref{tb:ARCHITECTURE}. We divide the dataset into three parts. From the first part, we obain the covariance matrix that we use for PCA and to obtain the PBA compressor. The second and third part are used as training and testing data for the purpose of classification. For a fixed rate, reconstructions are passed to the neural networks for training and testing respectively. Since our goal is to compare classification accuracy across the compressors, we fix both the architecture and hyperparameters, and do not perform any additional tuning for the separate algorithms.

Figure~\ref{fig:accuracy} shows that PBA outperforms PCA in terms of accuracy. The difference is especially significant for low rates; all algorithms attain roughly the same performance at higher rates.


\begin{table}[hbt]
\centering

\begin{tabular}[t]{lccc}
\hline
Hyperparameter&MNIST&CIFAR-10\\
\hline
Architecture & 2-layer fully connected NN & Convolutional Neural Network \\ & & with 2 convolutional layers, pooling and \\ & & three fully connected layers \\
\# Hidden Neurons&100& NA\\
Optimization Algorithm&Adam& SGD with momentum\\
Loss&Cross-entropy&Cross-entropy\\
Learning Rate&0.0005&0.01\\
\hline
\end{tabular}
\caption{Hyperparameter Choices and Architecture for Classification}\label{tb:ARCHITECTURE}
\end{table}%

\begin{figure}[hbt]
\centering
\minipage{0.45\textwidth}
  \includegraphics[width=\linewidth]{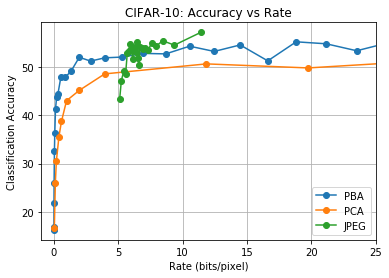}
\endminipage\hfill
\minipage{0.45\textwidth}
  \includegraphics[width=\linewidth]{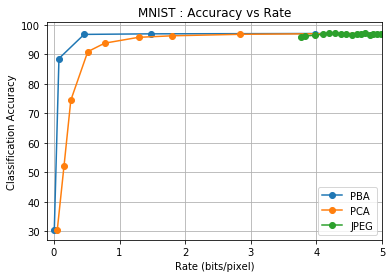}
\endminipage
\caption{Accuracy vs Rate (bits/pixel) for MNIST, CIFAR-10}\label{fig:accuracy}
\end{figure}

\clearpage
  \bibliographystyle{alpha}
  \bibliography{biblio.bib}
  \appendix
  \section{Review of Schur-Convexity}\label{app:schur}
    
  In this section, we review the key definitions and theorems related to Schur-convexity that we use in the proof of Theorem~\ref{thm:diagisoptimal}. 
  
  \begin{definition}{(\textbf{Majorization})}\cite{HornJ85}
For a vector $\bm{v} \in \RR^d$, let $\bm{v}^{\downarrow}$ denote the vector with the same components arranged in descending order. Given vectors $\bm{a}, \bm{b} \in \RR^{\ipdim}$, we say $\bm{a}$ majorizes $\bm{b}$ and denote $\bm{a} \succ \bm{b}$, if 
\[\sum\limits_{i=1}^{d} \left[\bm{a} \right]_i = \sum\limits_{i=1}^{d} \left[\bm{b} \right]_i,  \]
and for  all $k \in \left[ d-1 \right]$,
\[ \sum\limits_{i=1}^{k} \left[\bm{a}^{\downarrow} \right]_i \geq \sum\limits_{i=1}^{k} \left[\bm{b}^{\downarrow} \right]_i.\]

  \end{definition}

\begin{definition}{(\textbf{Schur-convexity})}
A function $f: \RR^{\ipdim} \rightarrow \RR$ is Schur-convex if for any vectors $\bm{a}, \bm{b} \in \RR^d$, such that $\bm{a} \succ \bm{b}$, 
\[ f\left( \bm{a} \right) \geq f\left( \bm{b} \right). \]
$f$ is strictly Schur-convex if the above inequality is a strict inequality for any $\bm{a} \succ \bm{b}$ that are not permutations of each other. $f$ is Schur-concave if the direction of the inequality is reversed and is strictly Schur concave if the direction of the inequality is reversed and it is a strict inequality for any $\bm{a} \succ \bm{b}$ that are not permutations of each other.
\end{definition}

\begin{proposition}\cite{MarshallOA11}\label{prop:single-letter}
If $f:\RR \rightarrow \RR$ is convex, then $\phi:\RR^d \rightarrow \RR$ given by
\[ \phi\left( \bm{v} \right) = \sum\limits_{i=1}^{d} f\left(\left[v \right]_i \right) \] 
is Schur-convex. If $f$ is concave, then $\phi$ is Schur-concave. Likewise if $f$ is strictly convex, $\phi$ is strictly Schur-convex and if $f$ is strictly concave, $\phi$ is strictly Schur-concave.
\end{proposition}

\end{document}